\documentclass{article}

\usepackage{PRIMEarxiv}

\usepackage[utf8]{inputenc}
\usepackage[T1]{fontenc}
\usepackage{hyperref}
\usepackage{url}
\usepackage{amsmath,amssymb,amsfonts,mathtools,amsthm}
\usepackage{bbm,blkarray}
\usepackage{graphicx}
\usepackage{booktabs,multirow,array,tabularx,longtable}
\usepackage{enumitem}
\usepackage{float}
\usepackage{pdflscape}
\usepackage{makeidx}
\usepackage{multicol}
\usepackage{comment}
\usepackage{adjustbox}
\usepackage{dirtree}
\usepackage{todonotes}
\usepackage{microtype}
\usepackage{nicefrac}

\usepackage{algorithm}
\usepackage{algorithmic}

\usepackage[compatibility=false]{caption}
\usepackage{subcaption}
\usepackage{tikz}
\usepackage{pgfplots}
\usetikzlibrary{calc,angles,quotes,arrows.meta,fit,positioning,backgrounds,trees}
\pgfplotsset{compat=1.18}

\usepackage[pageref]{backref}
\graphicspath{{./}}

\numberwithin{equation}{section}
\newtheorem{theorem}{Theorem}[section]

\newtheorem{lemma}[theorem]{Lemma}
\newtheorem{proposition}[theorem]{Proposition}
\newtheorem{corollary}[theorem]{Corollary}

\newtheorem{observation}[theorem]{Observation}

\theoremstyle{definition}
\newtheorem{definition}[theorem]{Definition}
\newtheorem{assumption}[theorem]{Assumption}

\newtheorem{example}[theorem]{Example}

\theoremstyle{remark}
\newtheorem{remark}[theorem]{Remark}

\setlist[itemize]{nosep,leftmargin=*,topsep=0pt,partopsep=0pt}

\makeindex

\renewcommand*{\backrefalt}[4]{%
\ifcase #1 %
(Not cited)%
\or
(Cited on p.~#2)%
\else
(Cited on pp.~#2)%
\fi
}

\title{Multiscale Exit--Join Dynamics: Tactical Consensus and Strategic Coalition Formation}

\author{
Quanyan Zhu\\
Department of Electrical and Computer Engineering\\
New York University Tandon School of Engineering\\
Brooklyn, NY, USA\\
\href{mailto:quanyan.zhu@nyu.edu}{\texttt{quanyan.zhu@nyu.edu}}
}

\begin{document}
\maketitle

\begin{abstract}
This paper develops a multiscale model of coalition formation in which
strategic exit--and--join decisions are coupled with tactical consensus
dynamics inside coalitions. Coalition value is generated endogenously from
within-coalition information aggregation, while Aumann--Dr\`eze payoffs,
switching frictions, and acceptance rules govern strategic reconfiguration. The
analysis characterizes joint tactical--strategic equilibria, conditions for
tactical unanimity, and mechanisms that sustain segregation or polarization.
Numerical experiments illustrate a central instability--consensus paradox: low
or negative switching barriers can prevent strategic convergence while
promoting enough temporal mixing to drive global tactical consensus.
\end{abstract}


\section{Introduction}

Coalition processes are inherently multiscale. Agents form, leave, and
reconfigure coalitions at a strategic level, while the members of each coalition
execute tasks, exchange information, and coordinate at a tactical level. These
two layers are coupled: coalition structure determines who interacts with whom,
and the outcomes of interaction reshape the incentives for future coalition
formation. The result is a strategic--tactical dynamical architecture in which
structure and performance co-evolve
\cite{aumann1974cooperative,fele2017coalitional,chalkiadakis2004bayesian,hamed2023distributed}.
At the strategic level, agents decide whether to remain in their current
coalition, exit, or join another coalition, and these decisions reshape the
interaction network \cite{aumann1974cooperative,hart1989potential}. At the
tactical level, coalition members repeatedly communicate, aggregate
information, and often reach internal agreement through consensus-like dynamics
\cite{hegselmann2015opinion,proskurnikov2015opinion}. Strategic
reconfiguration therefore changes the tactical interaction pattern, while
tactical outcomes change the strategic payoff landscape
\cite{hamed2023distributed,fele2017coalitional}.

This coupling appears in many multi-agent systems. In robotic teams, autonomous
robots form coalitions for exploration, surveillance, transportation, or other
tasks, and distributed control protocols coordinate behavior within each team
\cite{fele2017coalitional}. In multi-agent AI systems, agents may collaborate
temporarily for distributed inference or learning and then reorganize
partnerships according to the utility of previous collaborations
\cite{chalkiadakis2004bayesian,hamed2023distributed}. In social networks,
individuals form professional circles, online communities, or friendship
groups; repeated interaction within these groups can produce shared views, while
restricted cross-group interaction can sustain echo chambers and polarization
\cite{hegselmann2015opinion,proskurnikov2015opinion,del2017modeling,baumann2020modeling}.
Segregated interaction patterns can be measured using network-based segregation
indices \cite{bojanowski2014measuring}, and large-scale polarization dynamics
have been documented empirically \cite{levin2021dynamics}.

Movement between coalitions is rarely frictionless. Switching can incur
explicit costs, such as time, resources, renegotiation, or institutional
barriers. It can also incur implicit costs, such as cognitive discomfort, social
resistance, or reluctance to leave familiar environments. Behavioral mechanisms
such as confirmation bias reinforce these frictions \cite{del2017modeling}. We
refer to these impediments collectively as \emph{coalition barriers}. Such
barriers regulate the amount of mixing between groups and therefore shape the
global connectivity of the interaction network over time
\cite{baumann2020modeling,bojanowski2014measuring}.

The exit--and--join dynamics studied in this paper provide the elementary
strategic move through which coalition structure changes. At a strategic date,
an agent identifies its current coalition as the origin coalition and considers
joining either another existing coalition or an empty destination, the latter
corresponding to the formation of a singleton. This candidate move changes only
the origin and destination coalitions: the agent is removed from the origin and
inserted into the destination, while all other coalitions remain fixed. The move
is implemented only when it passes two gates. First, the agent's Aumann--Dr\`eze
payoff after the move, net of switching costs, must strictly improve. Second,
the destination coalition must accept the entrant. Thus, the process is neither
unrestricted mobility nor centralized coalition optimization; it is a local,
incentive-driven reconfiguration rule constrained by frictions and acceptance.

\subsection{Poiesis and Praxis in Coalition Systems}

The strategic--tactical framework admits a natural interpretation through
Aristotle's distinction between \emph{poiesis} and \emph{praxis}. In the
\emph{Nicomachean Ethics}, poiesis denotes productive activity oriented toward
bringing something into being, whereas praxis denotes action whose end lies in
the activity itself \cite{gottlieb2015aristotle}. Coalition formation
corresponds to poiesis. At the strategic level, agents
constitute, dissolve, and reconfigure coalitions. This activity produces the
interaction structure within which collective behavior occurs; it is structural
and generative because it determines the topology of relations
\cite{aumann1974cooperative}. Tactical execution corresponds to praxis. Within
a formed coalition, agents
interact, deliberate, exchange information, and coordinate. Consensus formation
and collaborative task performance unfold inside the structure that has been
produced \cite{hegselmann2015opinion,proskurnikov2015opinion}. In this layer,
the activity of interaction and coordination is itself the locus of realization.
The two dimensions are inseparable. Poiesis determines who interacts with whom;
praxis determines how states, beliefs, or outputs evolve under that structure.
The outcomes of praxis feed back into subsequent poiesis, since performance and
agreement influence future coalition decisions
\cite{hart1989potential,hamed2023distributed}. Coalition systems therefore
exhibit a recurring cycle: structure is produced, action unfolds within it, and
action reshapes future structure.

\subsection{Related Work}

The strategic layer of the present model is rooted in cooperative game theory
and coalition formation. Classical value theory studies how a transferable
surplus should be allocated among agents, beginning with the Shapley value and
continuing through coalition-structure values such as the Aumann--Dr\`eze value
and the Owen value \cite{Shapley1953,aumann1974cooperative,owen1977values}.
Related stability notions appear in hedonic coalition formation, stable
partitions, and endogenous coalition-structure models
\cite{dreze1980hedonic,BogomolnaiaJackson2002,AptRadzik2009,ray1999theory}.
Dynamic coalition formation further emphasizes that coalition structures may be
the outcome of a sequential adjustment process rather than a static allocation
\cite{konishi2003coalition,apt2009coalition,filar2000dynamic,bauso2009robust,ZhuHan2026SplitMerge}.
In multi-agent systems, coalition methods have been developed for task
allocation, distributed problem solving, and coalitional control
\cite{shehory1998methods,rahwan2015coalition,fele2017coalitional}. These
literatures provide the strategic vocabulary used here, but they typically
treat the characteristic function or payoff primitives as exogenous. The
present paper instead derives coalition value from the tactical consensus
outcome generated inside each coalition.

The tactical layer builds on consensus and opinion dynamics. DeGroot averaging,
social-influence models, and distributed agreement protocols characterize how
agents aggregate information over a fixed or time-varying interaction graph
\cite{degroot1974reaching,friedkin1990social,jadbabaie2003coordination,moreau2005stability,olfati2007consensus}.
Work on bounded confidence, antagonistic interactions, echo chambers, and
polarization shows how network constraints and selective interaction can
produce fragmentation rather than unanimity
\cite{hegselmann2002opinion,proskurnikov2015opinion,altafini2013consensus,del2017modeling,baumann2020modeling}.
In contrast with models that take the communication graph as primitive, the
framework here makes the graph endogenous: strategic exit--and--join decisions
determine which agents interact tactically, while tactical consensus changes
the payoffs that drive future reconfiguration. The closest methodological
connection is therefore to work on dynamic and
learning-based coalitional games, where incentives, state evolution, and
allocation rules interact over time
\cite{smyrnakis2019game,hamed2023distributed}. The distinction is that our
state variable is not only a game payoff or resource level, but also a slow
opinion profile shaped by repeated within-coalition averaging. This yields a
multiscale mechanism in which strategic instability can either obstruct
coalition convergence or, through temporal mixing, support tactical consensus.

\subsection{Contributions}

This work asks what global patterns emerge when strategic coalition formation is
coupled with tactical within-coalition consensus dynamics. Depending on
connectivity patterns and barrier levels, the system may converge to global
unanimity, stabilize in segregated structures, or sustain polarization across
persistent clusters
\cite{proskurnikov2015opinion,baumann2020modeling}. The paper makes four main
contributions. First, it introduces a fast--slow model
in which transferable coalition value is generated endogenously by
within-coalition consensus. Second, it formulates state-dependent
exit--and--join dynamics with Aumann--Dr\`eze payoffs, switching frictions, and
acceptance constraints. Third, it characterizes tactical unanimity, strategic
stability, segregation, and polarization in terms of temporal connectivity,
coalition barriers, and the geometry of the performance landscape. Fourth, it
uses numerical experiments to illustrate a central phenomenon: strategic
instability can induce tactical unanimity.

The last point is the main conceptual message. When coalition barriers are low,
agents frequently reconfigure partnerships. Although the coalition structure may
remain strategically unstable, the union of interactions over time can become
connected, and repeated cross-group mixing can drive global consensus
\cite{hegselmann2015opinion}. Conversely, high barriers stabilize coalition
structure but can preserve disconnected interaction patterns, allowing
segregation or polarization to persist
\cite{bojanowski2014measuring,del2017modeling}. Coalition barriers therefore
become a design and control variable. Reducing
barriers promotes mixing and unanimity, while increasing barriers stabilizes
structure at the risk of fragmentation. By tuning incentives, mobility costs, or
communication constraints, one can influence the emergent
strategic--tactical dynamics \cite{fele2017coalitional,hamed2023distributed}.
The framework developed here provides a mathematical basis for analyzing these
multiscale coalition phenomena and for understanding how strategic decisions and
tactical consensus jointly determine long-run collective outcomes.

\section{Consensus--Induced Coalition Value with Separated Time Scales}
\label{sec:consensus_induced_value}

A central modeling step of our framework is to specify how the transferable
utility \(v(S)\) of a coalition \(S\subseteq N\) is generated endogenously from
internal (tactical-level) consensus dynamics, while coalition formation itself
evolves on a slower strategic time scale through exit--and--join deviations.
This fast--slow separation allows coalition value to emerge from information
aggregation and coordination processes internal to each coalition, while
preserving a decentralized noncooperative mechanism for coalition
reconfiguration.

\subsection{Tactical states and consensus representation}

\begin{definition}[Strategic and tactical states]
We explicitly distinguish two interacting time scales. On the strategic time
scale, \(\tau = 0,1,2,\dots\) indexes the moments at which coalition formation
decisions occur.
At each \(\tau\), the system is characterized by a coalition structure
\(
T(\tau)\in\Pi(N),
\)
where \(\Pi(N)\) denotes the set of all partitions of the agent set \(N\).
At this time scale, agents may unilaterally exit their current coalition and join
another coalition, subject to the exit--and--join rules. At strategic step
\(\tau\), internal consensus evolves over
\(K_\tau\in\mathbb{N}\) fast tactical rounds, indexed by
\(k=0,1,\dots,K_\tau\).
During this interval, the coalition structure \(T(\tau)\) is held fixed while
within-coalition opinion dynamics unfold independently across coalitions.
\end{definition}

Each agent \(i\in N\) possesses an intrinsic opinion (or internal state)
\(
x_i^0(\tau)\in\mathbb{R}^d,
\)
which represents the agent's private information, belief, preference, or local
estimate of a task-relevant quantity at strategic time \(\tau\). At the
beginning of the tactical phase associated with \(\tau\), these intrinsic
opinions initialize the within-coalition dynamics.
When an agent exits one coalition and joins another at strategic time
\(\tau+1\), its tactical state in the new coalition is initialized using the
updated intrinsic opinion \(x_i^0(\tau+1)\).
Thus, intrinsic opinions evolve slowly across strategic iterations, while
tactical states are transient.

Fix a strategic time \(\tau\) and a coalition \(S\in T(\tau)\).
Each agent \(i\in S\) holds a tactical opinion state
\(x_i(\tau,k)\in\mathbb{R}^d\), \(k=0,1,\dots,K_\tau\), and we define the
stacked coalition state \(x_S(\tau,k):=\mathrm{col}\bigl(x_i(\tau,k)\bigr)_{i\in S}
\in\mathbb{R}^{d|S|}\).

The interaction topology within coalition \(S\) is encoded by a row-stochastic
matrix \(P_S=[p_{ij}^S]\), with \(p_{ij}^S\ge0\) and
\(\sum_{j\in S}p_{ij}^S=1\), which is assumed fixed over the tactical window
\(\{0,\dots,K_\tau\}\).

The within-coalition consensus dynamics follow a DeGroot-type update:
\begin{equation}
\label{eq:degroot_tau_k}
x_i(\tau,k+1)
=
\sum_{j\in S} p_{ij}^S\, x_j(\tau,k),
\qquad i\in S,
\end{equation}
with initial condition
\begin{equation}
\label{eq:tactical_init}
x_i(\tau,0)=x_i^0(\tau),
\qquad i\in S.
\end{equation}
These dynamics evolve independently across coalitions in \(T(\tau)\).

\begin{assumption}[Primitive coalition interactions]
\label{ass:primitive_interactions}
For every coalition \(S\subseteq N\), the matrix \(P_S\) is
primitive.
\end{assumption}

\begin{proposition}[Coalition consensus representation]
\label{prop:coalition_consensus_representation}
Under Assumption~\ref{ass:primitive_interactions}, every coalition
\(S\subseteq N\) admits a unique stationary distribution
\(\pi^S\in\Delta^{|S|}\) satisfying \((\pi^S)^\top P_S=(\pi^S)^\top\) and
\(\sum_{i\in S}\pi_i^S=1\).
As \(K_\tau\to\infty\), the tactical dynamics converge to a
consensus state
\begin{equation}
\label{eq:consensus_limit}
\lim_{k\to\infty} x_i(\tau,k)
=
\bar x_S(\tau),
\qquad \forall i\in S,
\end{equation}
where
\begin{equation}
\label{eq:weighted_consensus}
\bar x_S(\tau)
=
\sum_{i\in S} \pi_i^S\, x_i^0(\tau).
\end{equation}

Thus, at the strategic level, each coalition \(S\) is represented by a single
consensus state \(\bar x_S(\tau)\in\mathbb{R}^d\).
\end{proposition}

\begin{proof}
Primitivity of \(P_S\) implies, by the Perron--Frobenius theorem, that \(1\) is
a simple eigenvalue and that the associated stationary distribution
\(\pi^S\) is unique. Moreover, \(P_S^k\to \mathbf 1(\pi^S)^\top\), and applying
this limit to the initial tactical state \(x_S(\tau,0)\) gives
\eqref{eq:consensus_limit}--\eqref{eq:weighted_consensus}.
\end{proof}

\subsection{Slow adaptation and induced coalition value}

At the end of the tactical phase, intrinsic opinions are updated according to a
slow adaptation rule that blends prior beliefs with coalition-level consensus.
For each agent \(i\in S\subseteq T(\tau)\), define
\begin{equation}
\label{eq:intrinsic_update}
x_i^0(\tau+1)
=
(1-\gamma_i)\,x_i^0(\tau)
+
\gamma_i\,\bar x_S(\tau),
\qquad
\gamma_i\in[0,1].
\end{equation}

This update captures learning, belief revision, or assimilation of collective
information. When \(\gamma_i=0\), agents retain fixed intrinsic opinions; when
\(\gamma_i=1\), intrinsic opinions fully reset to the coalition consensus.
The updated values \(x_i^0(\tau+1)\) serve as initial conditions for the tactical
dynamics at strategic time \(\tau+1\).

\begin{definition}[Consensus-induced characteristic function]
\label{def:consensus_induced_characteristic}
Coalition value is determined by the quality of the consensus outcome attained
within the coalition.
Let
\(
V:\mathbb{R}^d \to \mathbb{R}
\)
be a coalition performance functional.
We define the transferable-utility characteristic function
\begin{equation}
\label{eq:v_general_discrete}
v(S;\tau)
:=
V\!\left(\bar x_S(\tau)\right),
\qquad
v(\varnothing;\tau)=0.
\end{equation}
Hence, coalition value is generated endogenously by the steady-state outcome of
fast consensus dynamics, while slow updates of intrinsic opinions induce
intertemporal dependence across strategic steps.
\end{definition}

\section{State-Dependent Exit--and--Join Dynamics}
\label{sec:state_dependence_exit_join}

We now formalize the exit--and--join dynamics and characterize the induced
state dependence of coalition values and individual incentives.
Coalition-level consensus evolves endogenously as a fast process nested within
a slower strategic reconfiguration of coalition membership.
The resulting dynamics define a state-dependent cooperative game whose payoff
allocations are determined by the Aumann--Dr\`eze value.
The mechanism is local and sequential. At each strategic time, a single active
agent compares its current payoff with the payoff it would receive after moving
from its current coalition to a candidate destination. The destination may be an
existing coalition or the empty destination, which represents forming a
singleton. A candidate move first generates a trial coalition structure; it is
then evaluated using the consensus-induced characteristic function and the
Aumann--Dr\`eze payoff allocation in that trial structure. The move is executed
only if the entrant obtains a strict net payoff improvement after switching
costs and the destination coalition accepts the entrant. Once an implementable
move is executed, the next tactical phase recomputes coalition consensus under
the new membership structure.

\subsection{Strategic state, moves, and payoffs}

At each strategic time
\(
\tau = 0,1,2,\dots,
\)
the system state is the pair
\begin{equation}
\label{eq:system_state}
\mathcal{S}(\tau)
=
\bigl(T(\tau), \{\bar x_S(\tau)\}_{S\in T(\tau)}\bigr),
\end{equation}
where \(T(\tau)\in\Pi(N)\) is a coalition structure and
\(\bar x_S(\tau)\in\mathbb{R}^d\) denotes the steady-state consensus outcome of
coalition \(S\). For any agent \(i\in N\), let
\begin{equation}
\label{eq:current_coalition}
C_{T(\tau)}(i)\in T(\tau)
\end{equation}
denote the unique coalition containing agent \(i\) at time \(\tau\).
This coalition is the agent's origin coalition. An exit--and--join decision is
specified by choosing a destination coalition different from the origin, or by
choosing the empty destination when the agent forms a singleton.

\begin{definition}[Exit--and--join operator]
\label{def:exit_join_operator}
At strategic time \(\tau\), agent \(i\) may unilaterally select a destination
coalition
\begin{equation}
\label{eq:destination_set}
D \in \mathcal{D}_i\bigl(T(\tau)\bigr)
:=
T(\tau)\setminus\{C_{T(\tau)}(i)\}
\;\cup\;\{\varnothing\},
\end{equation}
where \(D=\varnothing\) corresponds to forming a singleton coalition.

An exit--and--join move by agent \(i\) from its current coalition
\(C=C_{T(\tau)}(i)\) to destination \(D\) induces a new coalition structure
\begin{equation}
\label{eq:exit_join_operator}
T(\tau+1)
=
F_i\bigl(T(\tau),D\bigr),
\end{equation}
where the exit--and--join operator \(F_i\) is defined by
\begin{equation}
\label{eq:exit_join_definition}
F_i(T,D)
:=
\bigl(T \setminus \{C,D\}\bigr)
\;\cup\;
\{C\setminus\{i\}\}
\;\cup\;
\{D\cup\{i\}\},
\end{equation}
with empty coalitions removed by convention.
All coalitions other than \(C\) and \(D\) remain unchanged.
\end{definition}

The operator \(F_i\) should be interpreted as a trial transition rather than as
an automatic update. It describes what the partition would become if the move
were carried out. The subsequent payoff and acceptance tests determine whether
this trial transition is actually implemented.

Coalition surplus at time \(\tau\) is generated by consensus outcomes according to
\begin{equation}
\label{eq:consensus_value}
v(S;\tau)
=
V\!\left(\bar x_S(\tau)\right),
\qquad
v(\varnothing;\tau)=0.
\end{equation}

Given a coalition structure \(T\), individual payoffs are allocated using the
Aumann--Dr\`eze value.
For each agent \(i\in N\),
\begin{equation}
\label{eq:ad_payoff}
\Omega_i\bigl(v(\cdot;\tau);T\bigr)
=
\phi_i\!\left(v_{|C_T(i)}(\cdot;\tau)\right),
\end{equation}
where \(\phi_i\) denotes the Shapley value and
\(v_{|C_T(i)}\) is the restriction of \(v\) to agent \(i\)'s current coalition.

Equivalently, if \(S=C_T(i)\),
\begin{equation}
\label{eq:ad_expansion}
\Omega_i
=
\sum_{R\subseteq S\setminus\{i\}}
\frac{|R|!\,(|S|-|R|-1)!}{|S|!}
\Bigl[
v(R\cup\{i\};\tau)-v(R;\tau)
\Bigr].
\end{equation}

Agent \(i\) evaluates a move to destination \(D\) according to the net payoff
\begin{equation}
\label{eq:exit_join_payoff}
u_i\bigl(D;T(\tau)\bigr)
=
\Omega_i\!\left(v(\cdot;\tau);F_i(T(\tau),D)\right)
-
c_i\!\left(T(\tau),D\right),
\end{equation}
where \(c_i(T,D)\ge 0\) is a switching cost.

The move is \emph{admissible} if
\begin{equation}
\label{eq:admissibility}
u_i\bigl(D;T(\tau)\bigr)
>
\Omega_i\!\left(v(\cdot;\tau);T(\tau)\right),
\end{equation}
and may additionally be subject to an acceptance constraint by the destination
coalition (e.g., Pareto non-decrease of incumbents' AD payoffs).
Admissibility is therefore the entrant's private incentive test. It asks
whether the move is profitable for the moving agent after the switching cost is
paid, before considering whether the destination coalition is willing to admit
the entrant.

\subsection{Acceptance, recomputation, and state dependence}
\label{subsec:acceptance_and_costs}

In addition to individual payoff improvement, exit--and--join moves may be
subject to acceptance by the destination coalition.
Acceptance rules capture institutional, technological, or social frictions that
prevent coalitions from being freely entered, even when unilateral incentives
signal improvement.

\begin{definition}[Acceptance correspondence]
\label{def:acceptance_correspondence}
For any coalition structure $T$ and destination coalition
$D \in T \cup \{\varnothing\}$, the acceptance correspondence
\(\mathcal A_D(T) \subseteq N\) specifies the set of agents whose entry into
$D$ is approved by its incumbent members. An exit--and--join move by agent $i$
to destination $D$ is \emph{accepted} if \(i \in \mathcal A_D(T)\).
\end{definition}

A canonical acceptance rule requires Pareto non-decrease of Aumann--Dr\`eze
payoffs for incumbents:
\begin{equation}
\label{eq:pareto_acceptance}
\Omega_j\!\left(v(\cdot;\tau);F_i(T,D)\right)
\ge
\Omega_j\!\left(v(\cdot;\tau);T\right),
\qquad
\forall j \in D.
\end{equation}
This rule ensures that no incumbent agent is made strictly worse off by admitting
the entrant.

More permissive or restrictive rules can be accommodated, including majority
approval, weighted consent, or threshold-based acceptance.
The singleton destination $D=\varnothing$ is assumed to accept all entrants by
convention.

Switching costs $c_i(T,D)$ capture frictions associated with leaving a coalition
and integrating into a new one.
These costs may represent coordination delays, renegotiation overhead,
loss of trust, or reinitialization of tactical processes.
We allow $c_i(T,D)$ to depend on both the origin and destination coalitions, and
to be heterogeneous across agents.

Switching costs affect admissibility but not acceptance.
In particular, even if a move is accepted by the destination coalition, it is
implemented only if the agent's net payoff strictly improves.
Conversely, a payoff-improving move may still be blocked if it fails the
acceptance rule.

\begin{definition}[Implementable move]
\label{def:implementable_move}
Combining payoff improvement, switching frictions, and acceptance, an
exit--and--join move by agent $i$ from $C$ to $D$ at time $\tau$ is
\emph{implementable} if and only if
\begin{equation}
\label{eq:implementable_move}
\begin{aligned}
&\Omega_i\!\left(v(\cdot;\tau);F_i(T(\tau),D)\right)
-
c_i(T(\tau),D)
>
\Omega_i\!\left(v(\cdot;\tau);T(\tau)\right),
\\[0.4em]
&\text{and}\qquad
i \in \mathcal A_D\bigl(T(\tau)\bigr).
\end{aligned}
\end{equation}
\end{definition}

This separation between individual incentives, switching frictions, and coalition
acceptance clarifies the strategic structure of exit--and--join dynamics and
allows the analysis to distinguish unilateral profitability from collective
admissibility.
In the terminology used below, an admissible move is individually profitable,
an accepted move passes the destination's entry rule, and an implementable move
satisfies both requirements. Only implementable moves change the coalition
structure.

Following an implementable exit--and--join move at strategic time $\tau$,
the coalition structure updates to $T(\tau+1)$.
The subsequent tactical phase (stage $\tau+1$) starts from the intrinsic
profile $x^0(\tau+1)$ produced by the reset map \eqref{eq:intrinsic_update}.
Thus, consensus at stage $\tau+1$ is computed from $x^0(\tau+1)$, not from
the previous intrinsic profile $x^0(\tau)$.

Specifically, let agent $i$ move from its origin coalition
$C=C_{T(\tau)}(i)$ to a destination coalition $D$ at time $\tau$.
Then $T(\tau+1)=F_i(T(\tau),D)$ and the coalitions whose membership changes are
\(C' := C\setminus\{i\}\) and \(D' := D\cup\{i\}\).
The coalition-wise consensus values at strategic time $\tau+1$ are given by
\begin{equation}
\label{eq:consensus_recompute_revised}
\begin{aligned}
\bar x_{C'}(\tau+1)
&=
\sum_{j\in C'}
\pi_j^{\,C'}\, x_j^0(\tau+1),\\[0.5em]
\bar x_{D'}(\tau+1)
&=
\sum_{j\in D'}
\pi_j^{\,D'}\, x_j^0(\tau+1),
\end{aligned}
\end{equation}
where $\pi^{\,C'}$ and $\pi^{\,D'}$ denote the stationary distributions
associated with the within-coalition interaction matrices for $C'$ and $D'$.

For every \emph{unaffected} coalition
$S \in T(\tau)\setminus\{C,D\}$,
the membership of $S$ is unchanged across the strategic transition,
that is, \(S\in T(\tau+1)\) and \(S\) contains the same set of agents as at
time \(\tau\).
Accordingly, the within-coalition interaction pattern governing
the DeGroot dynamics \eqref{eq:degroot_tau_k} remains the same.

However, the consensus \emph{value} at stage $\tau+1$
is computed from the updated intrinsic profile
$x^0(\tau+1)$ rather than $x^0(\tau)$.
Thus, even for unaffected coalitions,
\begin{equation}
\label{eq:consensus_unaffected_revised}
\bar x_S(\tau+1)
=
\lim_{k\to K_{\tau+1}}
x_i(\tau+1,k),
\qquad
\forall i\in S,
\end{equation}
where the tactical dynamics at stage $\tau+1$
are initialized according to
\(x_i(\tau+1,0)=x_i^0(\tau+1)\).

Consequently, an implementable exit--and--join move induces
a discrete transition in the strategic state
\(\mathcal S(\tau)=\bigl(T(\tau),x^0(\tau)\bigr)\mapsto
\mathcal S(\tau+1)=\bigl(T(\tau+1),x^0(\tau+1)\bigr)\),
and the tactical dynamics in stage $\tau+1$
evolve according to the DeGroot updates
\eqref{eq:degroot_tau_k}
with initial condition
\eqref{eq:tactical_init}
applied to the coalition structure $T(\tau+1)$.

In the present framework, coalition values are generated endogenously by
within-coalition consensus dynamics.
At strategic time \(\tau\), the induced transferable-utility characteristic
function is defined by
\begin{equation}
\label{eq:state_dependent_game}
v(S;\tau)
=
V\!\left(\bar x_S(\tau)\right),
\qquad S \subseteq N,
\end{equation}
where each coalition consensus state \(\bar x_S(\tau)\) is determined by the
current coalition structure \(T(\tau)\).

As a result, the cooperative game is intrinsically \emph{state dependent}:
different coalition structures induce different collections of consensus states,
and hence different characteristic functions, even though the underlying agent
set \(N\) remains fixed.
Coalition formation therefore gives rise to a sequence of cooperative games whose
payoff structure evolves endogenously with coalition membership through the
consensus mechanism.

Although all agents within a coalition \(S\) converge to a common tactical
consensus state \(\bar x_S(\tau)\), their strategic incentives are generally
heterogeneous.
Under the Aumann--Dr\`eze value, the payoff allocated to an agent
\(i \in S\) is determined by its expected marginal contribution across all
subcoalitions of \(S\).

For any subset \(Q \subseteq S \setminus \{i\}\), the marginal contribution of
agent \(i\) is given by
\begin{equation}
\label{eq:marginal_contribution}
v(Q \cup \{i\};\tau) - v(Q;\tau)
=
V\!\left(\bar x_{Q \cup \{i\}}(\tau)\right)
-
V\!\left(\bar x_Q(\tau)\right).
\end{equation}
This quantity captures how the intrinsic opinion of agent \(i\) alters the
consensus outcome of the subcoalition \(Q\), and thus quantifies the influence of
agent \(i\) on coalition-level performance.

Consequently, exit--and--join incentives are driven by comparative marginal
\emph{influence} on consensus outcomes across coalitions, rather than by
disagreement within a coalition.
Even when tactical consensus is fully achieved, strategic incentives for
coalition reconfiguration may persist due to asymmetries in agents'
contributions to collective performance.

\subsection{Algorithmic summary}
\label{subsec:algorithmic_summary}

Algorithm~\ref{alg:exit_join_ad_multiscale} summarizes the coupled evolution of
coalition structure, tactical consensus, and individual incentives under the
proposed fast--slow framework.
The algorithm explicitly separates \emph{tactical-level information aggregation}
from \emph{strategic-level coalition reconfiguration}, while allowing slow
intertemporal learning through updates of intrinsic opinions.
Strategic transitions are governed jointly by individual payoff improvement,
switching frictions, and coalition acceptance constraints.

\begin{algorithm}[p]
\caption{Consensus-Induced Exit--and--Join Dynamics with Tactical and Strategic Updates}
\label{alg:exit_join_ad_multiscale}
\footnotesize
\begin{algorithmic}[1]
\STATE Initialize coalition structure \(T(0)\)
\STATE Initialize intrinsic opinions \(\{x_i^0(0)\}_{i\in N}\)

\FOR{\(\tau = 0,1,2,\dots\)}

    \STATE \textit{Tactical-level consensus dynamics}
    \FOR{each coalition \(S \in T(\tau)\)}
        \STATE Initialize tactical states:
        \(x_i(\tau,0) \leftarrow x_i^0(\tau)\) for all \(i \in S\)
        \FOR{\(k = 0,1,\dots,K_\tau-1\)}
            \FOR{each \(i \in S\)}
                \STATE Update \(x_i(\tau,k+1)
                \leftarrow \sum_{j\in S} p_{ij}^S\,x_j(\tau,k)\)
            \ENDFOR
        \ENDFOR
        \STATE Compute \(\bar x_S(\tau)
        \leftarrow \sum_{i\in S} \pi_i^S\,x_i^0(\tau)\)
    \ENDFOR

    \STATE \textit{Strategic-level payoff evaluation}
    \STATE Compute coalition values
    \(v(S;\tau) = V(\bar x_S(\tau))\) for all \(S \in T(\tau)\)
    \STATE Compute Aumann--Dr\`eze payoffs
    \(\Omega_i(v(\cdot;\tau);T(\tau))\) for all \(i \in N\)

    \STATE \textit{Exit--and--join decision}
    \STATE Select an active agent \(i \in N\)
    \STATE Initialize \textsc{MoveFound} \(\leftarrow\) \textbf{false}

    \FOR{each destination \(D \in \mathcal{D}_i(T(\tau))\)}
        \STATE Compute net payoff
        \(u_i(D;T(\tau))\) via~\eqref{eq:exit_join_payoff}
        \IF{\(
        u_i(D;T(\tau)) >
        \Omega_i(v(\cdot;\tau);T(\tau))
        \ \textbf{and}\
        i \in \mathcal A_D(T(\tau))
        \)}
            \STATE Set \(D^\star \leftarrow D\)
            \STATE \textsc{MoveFound} \(\leftarrow\) \textbf{true}
            \STATE \textbf{break}
        \ENDIF
    \ENDFOR

    \STATE \textit{Strategic update}
    \IF{\textsc{MoveFound}}
        \STATE Update coalition structure:
        \(T(\tau+1) \leftarrow F_i(T(\tau),D^\star)\)
    \ELSE
        \STATE \textbf{terminate}
        \STATE \textit{Exit--and--join equilibrium reached}
    \ENDIF

    \STATE \textit{Slow intrinsic-opinion update}
    \FOR{each coalition \(S \in T(\tau+1)\)}
        \FOR{each agent \(i \in S\)}
            \STATE Update \(x_i^0(\tau+1)
            \leftarrow (1-\gamma_i)x_i^0(\tau)+\gamma_i\bar x_S(\tau)\)
        \ENDFOR
    \ENDFOR

\ENDFOR
\end{algorithmic}
\end{algorithm}

At the beginning of each strategic iteration \(\tau\), the coalition structure
\(T(\tau)\) and intrinsic opinions \(\{x_i^0(\tau)\}_{i\in N}\) are treated as
fixed.
The algorithm first executes a \emph{tactical consensus phase}, during which
within-coalition DeGroot dynamics aggregate private information and yield
coalition-level consensus states \(\{\bar x_S(\tau)\}_{S \in T(\tau)}\).
Once tactical consensus is formed, the algorithm proceeds to the
\emph{strategic evaluation phase}.
Coalition values are computed by applying the performance functional \(V\) to
each consensus outcome, and individual payoffs are allocated using the
Aumann--Dr\`eze value.
A single agent is then selected to consider a unilateral exit--and--join
deviation.
A deviation is implemented only if it yields a strictly higher net payoff after
accounting for switching costs and is accepted by the destination coalition.
If no such admissible and accepted deviation exists, the process terminates at an
exit--and--join equilibrium.
Following any strategic update, the algorithm performs a
\emph{slow intrinsic-opinion update}.
Each agent partially assimilates the coalition consensus into its intrinsic
opinion, introducing intertemporal dependence across coalition formations.
Through the repeated interaction of fast tactical consensus dynamics, strategic
incentive-driven coalition reconfiguration subject to acceptance constraints, and
slow belief adaptation, the algorithm provides an operational realization of
consensus-induced coalition formation under Aumann--Dr\`eze incentives.

\section{Time-Dependent Induced Exit--and--Join Game}
\label{subsec:induced_game_GT}

We now define the noncooperative game induced by coalition formation when
coalition surplus is generated endogenously through tactical consensus dynamics
and evolves jointly with intrinsic agent beliefs over strategic time.
While the overall coalition-formation process is dynamic, the induced game at
each strategic time is a static normal-form game whose payoff structure depends
on the prevailing strategic state.
The resulting formulation is therefore both \emph{state dependent} and
\emph{time indexed}.

\subsection{Induced game at a strategic state}

Let \(\tau \in \mathbb{N}\) denote the strategic time index.
At time \(\tau\), the strategic state of the system is given by
\(
\bigl(T(\tau), \{x_i^0(\tau)\}_{i\in N}\bigr),
\)
where \(T(\tau)\in\Pi(N)\) is a coalition structure and
\(x_i^0(\tau)\in\mathbb{R}^d\) is the intrinsic opinion of agent \(i\).

Conditional on \(T(\tau)\) and \(\{x_i^0(\tau)\}\), we assume that tactical
consensus dynamics within each coalition have fully converged.
Accordingly, each coalition \(S\in T(\tau)\) is represented by a consensus state
\(
\bar x_S(\tau)\in\mathbb{R}^d,
\)
defined as the unique fixed point of the within-coalition consensus dynamics.
This fixed point admits the weighted-average representation
\begin{equation}
\label{eq:GT_consensus_map}
\bar x_S(\tau)
=
\sum_{i\in S} \pi_i^S\, x_i^0(\tau),
\end{equation}
where \(\pi^S\in\Delta^{|S|}\) is the stationary distribution associated with the
interaction matrix of coalition \(S\).

The tactical consensus outcomes induce a transferable-utility characteristic
function at strategic time \(\tau\),
\begin{equation}
\label{eq:GT_characteristic}
v(S;\tau)
=
V\!\left(\bar x_S(\tau)\right),
\qquad
v(\varnothing;\tau)=0,
\end{equation}
where \(V:\mathbb{R}^d\to\mathbb{R}\) is a coalition performance functional.
Because the consensus states \(\bar x_S(\tau)\) depend on both the coalition
structure \(T(\tau)\) and intrinsic opinions \(\{x_i^0(\tau)\}\), the
characteristic function \(v(\cdot;\tau)\) evolves endogenously over strategic
time.

The player set is the agent set
\(
\mathcal N := N.
\)
For each agent \(i\in N\), let \(C_{T(\tau)}(i)\in T(\tau)\) denote the unique
coalition containing \(i\).
At coalition structure \(T(\tau)\), the admissible action set of agent \(i\) is
\begin{equation}
\label{eq:GT_action_set}
A_i\bigl(T(\tau)\bigr)
:=
\bigl(T(\tau)\setminus\{C_{T(\tau)}(i)\}\bigr)\;\cup\;\{\varnothing\}.
\end{equation}
An action \(D\in A_i(T(\tau))\) represents a unilateral exit--and--join decision,
with \(D=\varnothing\) corresponding to the formation of a singleton coalition.

Fix a strategic time $\tau$ and a coalition structure
$T(\tau) \in \Pi(N)$. For each agent $i \in N$, let
\(
C_{T(\tau)}(i) \in T(\tau)
\)
denote the unique coalition containing $i$, and let the admissible
destination set be
\(A_i\bigl(T(\tau)\bigr):=
\bigl(T(\tau)\setminus\{C_{T(\tau)}(i)\}\bigr)\cup\{\varnothing\}\).

For $D \in A_i\bigl(T(\tau)\bigr)$, the unilateral exit--and--join
operator is defined by
\begin{equation}
\label{eq:GT_exit_join_operator_clean}
F_i\bigl(T(\tau),D\bigr)
=
\bigl(T(\tau)\setminus\{C_{T(\tau)}(i),D\}\bigr)
\;\cup\;
\{C_{T(\tau)}(i)\setminus\{i\}\}
\;\cup\;
\{D\cup\{i\}\},
\end{equation}
with empty coalitions removed by convention.

Thus $F_i(T,D)$ denotes the coalition structure obtained when agent $i$
exits its current coalition $C_{T}(i)$ and joins the destination
coalition $D$. Only the origin and destination coalitions are modified;
all other coalitions remain unchanged.

The consensus state under a proposed deviation is updated locally. Let
$C := C_{T(\tau)}(i)$. Under the deviated structure
$T' = F_i(T(\tau),D)$, only the coalitions $C\setminus\{i\}$ and
$D\cup\{i\}$ require recomputation of their consensus states. These are
given by
\begin{equation}
\label{eq:GT_consensus_update_clean}
\begin{aligned}
\bar x_{C\setminus\{i\}}(\tau)
&=
\sum_{j\in C\setminus\{i\}}
\pi_j^{\,C\setminus\{i\}}\,
x_j^0(\tau),\\
\bar x_{D\cup\{i\}}(\tau)
&=
\sum_{j\in D\cup\{i\}}
\pi_j^{\,D\cup\{i\}}\,
x_j^0(\tau).
\end{aligned}
\end{equation}
For every other coalition $S \in T(\tau)$ with
$S \notin \{C,D\}$, the consensus state remains unchanged.

Let $\Omega_i(v(\cdot;\tau);T)$ denote the Aumann--Dr\`eze value of
agent $i$ under coalition structure $T$ for the characteristic
function $v(\cdot;\tau)$.

\begin{definition}[Time-dependent induced exit--and--join game]
\label{def:induced_exit_join_game}
Fix a strategic time $\tau$ and a coalition structure
$T(\tau) \in \Pi(N)$.

For each agent $i \in N$ and each destination
$D \in A_i(T(\tau))$, define the unilateral payoff
\begin{equation}
\label{eq:GT_payoff_clean}
u_i(D;T(\tau),\tau)
=
\Omega_i\!\left(
v(\cdot;\tau);
F_i\bigl(T(\tau),D\bigr)
\right)
-
c_i\bigl(T(\tau),D\bigr),
\end{equation}
where $c_i(T,D)\ge 0$ denotes a switching cost.

The induced exit--and--join game at time $\tau$ is the
normal-form game
\begin{equation}
\label{eq:GT_definition_clean}
G_\tau\bigl(T(\tau)\bigr)
=
\Bigl(
N,\{A_i(T(\tau))\}_{i\in N},
\{u_i(\cdot;T(\tau),\tau)\}_{i\in N}
\Bigr).
\end{equation}

Payoffs are evaluated under unilateral deviations:
only the move of agent $i$ is implemented, while all other agents'
coalition memberships remain fixed.
\end{definition}

\begin{remark}[Augmented payoffs, effort, and switching gains]
\label{rem:augmented_payoffs}
The payoff specification \eqref{eq:GT_payoff_clean} separates
coalitional surplus allocation (via the Aumann--Dr\`eze value)
from switching frictions. More generally, the payoff of agent $i$
may incorporate additional effort or operational terms associated
with participation in a coalition.

A convenient formulation is
\(u_i(D;T(\tau),\tau)=
\Omega_i\!\left(v(\cdot;\tau);F_i\bigl(T(\tau),D\bigr)\right)
-\kappa_i\bigl(T(\tau),D\bigr)\),
where the generalized transition term
\(
\kappa_i(T,D)
=
c_i(T,D) + e_i\!\left(F_i(T,D)\right)
\)
may include switching costs such as renegotiation, delay, or loss of trust;
internal effort costs within the destination coalition; coordination or
cognitive burden; and incentive payments.

Importantly, $\kappa_i(T,D)$ need not be nonnegative.
In some environments, switching may generate a net gain
due to external incentives, subsidies, improved technology,
or reduced effort requirements in the destination coalition.
In such cases, $\kappa_i(T,D) < 0$ represents a switching gain.

\end{remark}

Although each game $G_\tau(T(\tau))$ is static, the family
$\{G_\tau\}_{\tau\in\mathbb{N}}$ evolves endogenously as intrinsic
opinions, consensus states, and coalition structures change over time.
Coalition formation therefore unfolds as a sequence of
state-dependent, time-indexed noncooperative games coupled through
fast tactical consensus dynamics and slow belief adaptation.

\subsection{Equilibrium and fixed-point characterization}

The equilibrium concept must jointly account for tactical consensus formation,
strategic coalition stability, and slow belief adaptation across time scales.

\begin{definition}[Joint tactical--strategic equilibrium]\label{def:eqm}
A pair
\(
\bigl(T^\star,\{x_i^{0,\star}\}_{i\in N}\bigr)
\)
is called a \emph{joint tactical--strategic equilibrium} if the following
conditions hold.
For the coalition structure \(T^\star\), the within-coalition consensus
dynamics converge as the tactical horizon tends to infinity, yielding coalition
consensus states \(\bar x_{C_{T^\star}(i)}^\star\), where
\(\bar x_{C_{T^\star}(i)}^\star\) is the \(k\to\infty\) limit of
\(x_i(\tau,k)\) for each \(i\in N\), which are invariant across strategic
iterations.
Given the consensus states
\(\{\bar x_S^\star\}_{S\in T^\star}\), no agent has a profitable unilateral
exit--and--join deviation:
\begin{equation}
\label{eq:GT_joint_equilibrium}
u_i\bigl(C_{T^\star}(i);T^\star\bigr)
\ge
u_i(D;T^\star),
\qquad
\forall i\in N,\;\forall D\in A_i(T^\star).
\end{equation}

Intrinsic opinions are stationary under the slow belief-update dynamics,
\begin{equation}
\label{eq:intrinsic_fixed_point}
x_i^{0,\star}
=
(1-\gamma_i)x_i^{0,\star}
+
\gamma_i \bar x_{C_{T^\star}(i)}^\star,
\qquad \forall i\in N.
\end{equation}
For \(\gamma_i>0\), this condition is equivalent to
\begin{equation}
\label{eq:intrinsic_equals_consensus}
x_i^{0,\star}
=
\bar x_{C_{T^\star}(i)}^\star,
\qquad \forall i\in N.
\end{equation}
\end{definition}

We characterize joint tactical--strategic equilibria as fixed points
of a discrete-time operator defined on an augmented state space.

\begin{definition}[Composite tactical--strategic operator]
\label{def:composite_operator}
Let $\mathcal X := \Pi(N)\times(\mathbb{R}^d)^N$ denote the joint state
space. An element $(T,x^0)\in\mathcal X$ consists of a coalition
structure $T\in\Pi(N)$ and an intrinsic-opinion profile
$x^0=(x_i^0)_{i\in N}\in(\mathbb{R}^d)^N$.

The first component is the tactical consensus map
$\mathcal C:\mathcal X\to\mathcal X$, defined by
\(\mathcal C(T,x^0)=(T,\bar x)\),
where for each coalition $S\in T$ and every $i\in S$,
$\bar x_i=\bar x_S := \sum_{j\in S}\pi_j^S x_j^0$,
and $\pi^S$ denotes the stationary distribution of the primitive
interaction matrix associated with coalition $S$.
Thus $\mathcal C$ replaces intrinsic opinions by coalition-level
consensus states while leaving the partition unchanged.

Given $(T,\bar x)\in\mathcal X$, define coalition values
$v_T(S)=V(\bar x_S)$ for $S\subseteq N$ and compute the corresponding
Aumann--Dr\`eze payoffs.
The strategic exit--and--join map
$\mathcal S:\mathcal X\rightrightarrows\mathcal X$ denotes the
(possibly set-valued) strategic map selecting all states $(T',\bar x')$
such that $T'$ is obtained from $T$ via an admissible unilateral
exit--and--join deviation under payoffs induced by $v_T$, with
$\bar x'=\bar x$ whenever $T'=T$.
If $T$ admits no profitable unilateral deviation, then
$\mathcal S(T,\bar x)=\{(T,\bar x)\}$.
A selection rule (e.g., sequential activation) renders $\mathcal S$
single-valued.

Finally, the slow belief-update map
$\mathcal B:\mathcal X\to\mathcal X$ is defined by
\(\mathcal B(T,x)=(T,x^+)\), where
\(x_i^{+}=(1-\gamma_i)x_i+\gamma_i x_{C_T(i)}\),
for all $i\in N$, with $\gamma_i\in(0,1]$ and $x_{C_T(i)}$ denoting
the common consensus value of the coalition containing $i$.

Fix a selection rule for $\mathcal S$. The one-step evolution of the
system is described by the composite operator
\(\Phi := \mathcal B\circ\mathcal S\circ\mathcal C:\mathcal X\to\mathcal X\),
so that $(T(\tau),x^0(\tau))$ evolves according to
$(T(\tau+1),x^0(\tau+1))=\Phi(T(\tau),x^0(\tau))$.
\end{definition}

\begin{proposition}[Fixed-point characterization]
\label{prop:joint_equilibrium_fixed_point}
A state $(T^\star,x^{0,\star})\in\mathcal X$ is a
joint tactical--strategic equilibrium if and only if it satisfies
$\Phi(T^\star,x^{0,\star})=(T^\star,x^{0,\star})$.
\end{proposition}

\begin{proof}
By Definition~\ref{def:composite_operator}, the map \(\Phi\) first applies the
tactical consensus map, then the strategic exit--and--join map, and finally the
slow belief-update map. A fixed point therefore has tactical consistency,
admits no profitable unilateral exit--and--join deviation, and is invariant
under the belief-update rule. These are exactly the requirements in
Definition~\ref{def:eqm}. Conversely, any state satisfying
Definition~\ref{def:eqm} is unchanged by each component of \(\Phi\), and hence
is fixed by their composition.
\end{proof}

At such a state, \(x_i^{0,\star}=\bar x_{C_{T^\star}(i)}\) for all \(i\),
\(T^\star\) admits no profitable unilateral exit--and--join deviation, and
intrinsic opinions are invariant under the belief-update rule.
Hence no further evolution occurs at either time scale.

\subsection{Existence}
\label{subsec:existence_joint_equilibrium}

\begin{assumption}[Finite-improvement setting]
\label{ass:finite_improvement_setting}
The agent set $N$ is finite;
for every coalition $S\subseteq N$, the interaction matrix $P_S$ is
row-stochastic, irreducible, and aperiodic; switching costs satisfy
$c_i(T,D)\ge 0$ for all $i,T,D$; belief-update gains satisfy
$\gamma_i\in(0,1]$ for all $i\in N$; and, for every fixed
intrinsic-opinion profile, the directed graph of implementable strict
exit--and--join improvements on \(\Pi(N)\) is acyclic. Equivalently, the
strategic improvement process has the finite-improvement property.
\end{assumption}

We now give a sufficient condition for the composite operator \(\Phi\) to admit
a fixed point.

\begin{remark}[Consensus convergence]
If $P_S$ is row-stochastic, irreducible, and aperiodic, then by the
Perron--Frobenius theorem the eigenvalue $1$ is simple and there exists
a unique stationary distribution $\pi^S\in\Delta^{|S|}$ satisfying
$(\pi^S)^\top P_S = (\pi^S)^\top$.
Moreover, the DeGroot dynamics satisfy
$P_S^k \to \mathbf 1 (\pi^S)^\top$ as $k\to\infty$,
so opinions within coalition $S$ converge to the unique consensus
$\bar x_S = \sum_{i\in S}\pi_i^S x_i^0$.
\end{remark}

\begin{lemma}
\label{lem:tactical_existence}
For any coalition structure \(T\in\Pi(N)\) and intrinsic-opinion profile
\(\{x_i^0\}_{i\in N}\), the tactical consensus map \(\mathcal C\) is uniquely
defined.
\end{lemma}

\begin{proof}
By primitivity of \(P_S\), the DeGroot consensus dynamics within each coalition
\(S\) converge to a unique stationary distribution \(\pi^S\).
Hence the consensus state \(\bar x_S\) exists and is unique for all
\(S\in T\). \qedhere
\end{proof}
\begin{lemma}
\label{lem:strategic_existence}
For any fixed intrinsic-opinion profile \(\{x_i^0\}_{i\in N}\), there exists at
least one coalition structure admitting no implementable profitable unilateral
exit--and--join deviation.
\end{lemma}

\begin{proof}
The set \(\Pi(N)\) of coalition structures is finite.
Define a directed graph on \(\Pi(N)\) by introducing an edge
\(T\to T^{i\to D}\) whenever agent \(i\) has an implementable strictly profitable
exit--and--join deviation from \(T\).
By the finite-improvement property, this finite directed graph is acyclic.
Every finite acyclic directed graph contains at least one sink node, and any such
sink corresponds to a coalition structure with no implementable profitable
unilateral deviation. \qedhere
\end{proof}
\begin{lemma}
\label{lem:belief_existence}
For any coalition structure \(T\in\Pi(N)\), the belief-update map admits a
nonempty set of fixed points characterized by
\(x_i^{0,\star}=\bar x_{C_T(i)}\) for all \(i\in N\).
Equivalently, intrinsic opinions are constant on each coalition of \(T\).
\end{lemma}

\begin{proof}
For each agent \(i\), fixed-point consistency of
\eqref{eq:intrinsic_update} gives
\(\gamma_i\bigl(\bar x_{C_T(i)}-x_i^{0,\star}\bigr)=0\).
Since \(\gamma_i>0\), this is equivalent to
\(x_i^{0,\star}=\bar x_{C_T(i)}\). Conversely, any profile that is constant on
each coalition satisfies this equality and is therefore fixed by the
belief-update map. Such profiles exist, for example by assigning an arbitrary
vector in \(\mathbb R^d\) to each coalition. \qedhere
\end{proof}
\begin{theorem}[Existence under finite improvement]
\label{thm:existence_joint_equilibrium}
Under Assumption~\ref{ass:finite_improvement_setting}, there exists at least one joint
tactical--strategic equilibrium.
\end{theorem}

\begin{proof}
Fix any vector \(z\in\mathbb R^d\) and set
\(x_i^{0,\star}=z\) for all \(i\in N\).
For every coalition structure \(T\), this profile is already tactically
consistent: each coalition consensus is \(\bar x_S=z\), so the tactical
consensus map and the belief-update map both leave the profile unchanged.

Apply Lemma~\ref{lem:strategic_existence} to this fixed intrinsic-opinion
profile. There exists a coalition structure \(T^\star\in\Pi(N)\) with no
implementable profitable unilateral exit--and--join deviation. Therefore the
strategic map leaves \(T^\star\) unchanged, while the tactical and belief maps
leave \(\{x_i^{0,\star}\}_{i\in N}\) unchanged. Hence this state is a fixed
point of \(\Phi\).
By Definition~\ref{def:eqm}, this fixed point is a joint
tactical--strategic equilibrium.
\end{proof}

\begin{remark}
The existence result relies only on finiteness of the agent set, convergence of
tactical consensus dynamics, positive belief assimilation, and the
finite-improvement property at the strategic level. No convexity or
superadditivity assumptions on the coalition performance functional \(V\) are
required. Convexity becomes relevant only for uniqueness, efficiency, or global
convergence properties.
\end{remark}

\section{Emergent Behavior}

\subsection{Tactical and Strategic Unanimity}

Let $(T^\star,\{x_i^{0,\star}\}_{i\in N})$
be a joint tactical--strategic equilibrium.
For each coalition $S \in T^\star$, define its consensus value
\(\bar x_S^\star=\sum_{i\in S}\pi_i^S x_i^{0,\star}\).

\begin{definition}[Tactical unanimity]
The equilibrium exhibits \emph{tactical unanimity} if
\(x_i^{0,\star}=\bar x_S^\star\) for all \(S\in T^\star\) and \(i\in S\).
That is, tactical states coincide within every coalition.
\end{definition}

\begin{definition}[Strategic unanimity]
The equilibrium exhibits \emph{strategic unanimity} if
\(
T^\star = \{N\}.
\)
That is, all agents form the grand coalition.
\end{definition}

\subsubsection{Strategic Unanimity}
Consider the frictionless case with zero switching costs and automatic
acceptance. If the grand coalition $N$ is strategically stable, then no
agent $i\in N$ prefers to exit and form a singleton. Hence
$\Omega_i(v;N)\ge v(\{i\})$ for all $i$.

Under the Shapley allocation, $\Omega_i(v;N)$ is a convex combination
of marginal contributions $v(S\cup\{i\})-v(S)$ over
$S\subseteq N\setminus\{i\}$, with strictly positive weights.
In particular,
\(\Omega_i(v;N)\le v(N)-v(N\setminus\{i\})\). Therefore stability of $N$ requires
\(v(N)\ge v(N\setminus\{i\})+v(\{i\})\) for all \(i\in N\).
Equivalently, each agent's marginal contribution to the grand coalition
must be at least as large as its standalone value.

\begin{theorem}[Convexity and strategic unanimity]
\label{thm:convex_dominance_unanimity}
Let $N$ be finite. Define
\(\bar x_S=|S|^{-1}\sum_{j\in S}x_j^0\) and
\(v(S)=V(\bar x_S)-\kappa(|S|)\), where
$V:\mathbb{R}\to\mathbb{R}$ is continuously differentiable and
$\kappa:\mathbb{N}\to\mathbb{R}$ is convex.
Define
\(M:=\sup_x|V'(x)|\) and \(D:=\max_{i,j}|x_i^0-x_j^0|\). Assume that for all $s<t$,
\begin{equation}
\label{eq:convex_dominance_condition}
\big[\kappa(t+1)-\kappa(t)\big]
-
\big[\kappa(s+1)-\kappa(s)\big]
\;\le\;
M D
\left(
\frac{1}{s+1}
+
\frac{1}{t+1}
\right).
\end{equation}
Then $v$ is supermodular.
If, in addition, $V$ is strictly concave,
switching costs are zero, and acceptance is automatic,
then the grand coalition $\{N\}$ is strategically stable.
\end{theorem}

\begin{proof}
Fix $S\subseteq T\subseteq N$ and $i\notin T$.
Let $s=|S|$ and $t=|T|$. Write
\(\Delta_S^V=V(\bar x_{S\cup\{i\}})-V(\bar x_S)\) and
\(\Delta_T^V=V(\bar x_{T\cup\{i\}})-V(\bar x_T)\). By the mean value theorem,
\(\Delta_S^V=\frac{1}{s+1}V'(\xi_S)(x_i^0-\bar x_S)\),
for some $\xi_S$ between $\bar x_S$ and $\bar x_{S\cup\{i\}}$.
Hence
\(|\Delta_S^V|\le\frac{M}{s+1}|x_i^0-\bar x_S|\le\frac{MD}{s+1}\).
Similarly,
\(|\Delta_T^V|\le\frac{MD}{t+1}\). Therefore,
\(|\Delta_S^V-\Delta_T^V|\le
MD\!\left(\frac{1}{s+1}+\frac{1}{t+1}\right)\). Compute the discrete marginal difference:
\[
\begin{aligned}
&\big[v(S\cup\{i\})-v(S)\big]
-
\big[v(T\cup\{i\})-v(T)\big]
\\
&=
\Delta_S^V-\Delta_T^V
-
\Big(
\kappa(s+1)-\kappa(s)
-
\kappa(t+1)+\kappa(t)
\Big).
\end{aligned}
\]
Rewriting,
this equals
\(\Delta_S^V-\Delta_T^V+
\big(\kappa(t+1)-\kappa(t)-\kappa(s+1)+\kappa(s)\big)\).
By the dominance condition \eqref{eq:convex_dominance_condition},
\(\kappa(t+1)-\kappa(t)-\kappa(s+1)+\kappa(s)
\ge -MD\left(\frac{1}{s+1}+\frac{1}{t+1}\right)\).
Since
\(\Delta_S^V-\Delta_T^V\ge
-MD\left(\frac{1}{s+1}+\frac{1}{t+1}\right)\),
we obtain
\(\big[v(S\cup\{i\})-v(S)\big]-\big[v(T\cup\{i\})-v(T)\big]\ge0\).
Thus $v$ is supermodular. Assume now that $V$ is strictly concave.
Then aggregation strictly increases total surplus unless all
intrinsic opinions coincide.
Hence the grand coalition uniquely maximizes total surplus. Supermodularity implies increasing marginal contributions:
for all $S\subseteq N\setminus\{i\}$,
\(v(S\cup\{i\})-v(S)\le v(N)-v(N\setminus\{i\})\).
In particular, \(v(\{i\})\le v(S\cup\{i\})-v(S)\). The Shapley payoff is
\[
\Omega_i(v;N)
=
\sum_{S\subseteq N\setminus\{i\}} w_S
\big(v(S\cup\{i\})-v(S)\big),
\quad w_S>0.
\]
Since every marginal term is at least $v(\{i\})$,
\(\Omega_i(v;N)\ge v(\{i\})\). Under zero switching costs and automatic acceptance,
no agent benefits from deviating to singleton.
Hence the grand coalition is strategically stable.
\end{proof}

\subsubsection{Tactical Unanimity}

The tactical dynamics defined in \eqref{eq:degroot_tau_k}--\eqref{eq:intrinsic_update}
can be embedded into a single global switching-system representation. Recall that at each strategic time $\tau$ the tactical dynamics
\eqref{eq:degroot_tau_k} evolve for $k=0,\dots,K_\tau-1$.
Define cumulative switching times
\(T_0:=0\) and \(T_{\tau+1}:=T_\tau+K_\tau\).
Introduce a global discrete time index
\(
t = 0,1,2,\dots
\)
such that the tactical phase corresponding to strategic time $\tau$
occupies the interval
\(
t \in \{T_\tau, \dots, T_{\tau+1}-1\}.
\) Let
\(
x(t) := \operatorname{col}\bigl(x_i(t)\bigr)_{i\in N}
\in \mathbb{R}^{dn}.
\) 
At the switching instant $t=T_\tau$, we identify
\(
x(T_\tau)=x^0(\tau),
\)
where $x^0(\tau)$ is the intrinsic profile defined in
\eqref{eq:tactical_init}.

Fix a strategic time $\tau$ and let $T(\tau)\in\Pi(N)$ be the
current coalition structure. Recall that for each coalition
$S\in T(\tau)$, the within-coalition tactical dynamics
are governed by the DeGroot update
\eqref{eq:degroot_tau_k} with interaction matrix
\(P_S=[p_{ij}^S]_{i,j\in S}\), where \(p_{ij}^S\ge0\) and
\(\sum_{j\in S}p_{ij}^S=1\).
Thus $P_S$ is a row-stochastic matrix defined on the index set $S$.
The coalition structure $T(\tau)$ induces a block-diagonal matrix
\[
\widetilde P(\tau)
=
\operatorname{blkdiag}\bigl(P_S\bigr)_{S\in T(\tau)}
\in\mathbb{R}^{n\times n},
\]
where each block $P_S$ occupies the rows and columns
corresponding to the agents in $S$, and all off-block entries are zero.

Equivalently, for $i,j\in N$,
\[
\widetilde p_{ij}(\tau)
=
\begin{cases}
p_{ij}^S, & \text{if } i,j\in S \text{ for some } S\in T(\tau),\\
0, & \text{otherwise}.
\end{cases}
\]

Because each block $P_S$ is row-stochastic,
$\widetilde P(\tau)$ is also row-stochastic:
\(\widetilde p_{ij}(\tau)\ge0\) and
\(\sum_{j\in N}\widetilde p_{ij}(\tau)=1\) for all \(i\in N\).
Hence $\widetilde P(\tau)\mathbf 1=\mathbf 1$,
where $\mathbf 1$ denotes the all-ones vector in $\mathbb{R}^n$.

The associated directed graph
$\mathcal G(\tau)=(N,E(\tau))$ has edge $(j,i)$ whenever
$\widetilde p_{ij}(\tau)>0$.
Then $\mathcal G(\tau)$ is the disjoint union of the
subgraphs induced by the coalitions in $T(\tau)$.
There are no inter-coalition edges, reflecting the fact that
tactical interaction occurs only within coalitions.

If each $P_S$ is primitive,
then for each coalition $S$ there exists a unique stationary
distribution $\pi^S\in\Delta^{|S|}$ satisfying
\(
(\pi^S)^\top P_S = (\pi^S)^\top.
\)
The spectrum of $\widetilde P(\tau)$ is the union of the spectra
of its blocks.
In particular, the eigenvalue $1$ has multiplicity equal to
$|T(\tau)|$, and the corresponding right eigenspace is spanned by
vectors that are constant on each coalition.
Consequently,
\[
\widetilde P(\tau)^k
\;\longrightarrow\;
\operatorname{blkdiag}
\bigl(
\mathbf 1 \, (\pi^S)^\top
\bigr)_{S\in T(\tau)}
\quad \text{as } k\to\infty,
\]
which recovers the coalition-wise consensus limit
\eqref{eq:consensus_limit}.

The switching signal is
\(
\sigma(t)=\tau
\quad
\text{whenever }
t\in[T_\tau,T_{\tau+1}-1].
\) Then \eqref{eq:degroot_tau_k} can be written globally as the
switching linear system
\begin{equation}
\label{eq:global_switching}
x(t+1)
=
\widetilde P(\sigma(t))\,x(t).
\end{equation}

For $t\in[T_\tau,T_{\tau+1}-1]$,
\(
x(t)
=
\widetilde P(\tau)^{\,t-T_\tau} x^0(\tau).
\)
Under the primitivity assumption on each $P_S$,
the limit \eqref{eq:consensus_limit} becomes
\begin{equation}
\label{eq:global_consensus}
\lim_{t\uparrow T_{\tau+1}}
x(t)
=
\bar x(\tau),
\end{equation}
where $\bar x(\tau)$ stacks the coalition consensus values
\(\bar x_S(\tau)\) defined in \eqref{eq:consensus_limit}.

At $t=T_{\tau+1}$, intrinsic opinions are updated according to
\eqref{eq:intrinsic_update}, which in stacked form reads
\begin{equation}
\label{eq:global_intrinsic_update}
x^0(\tau+1)
=
(I-\Gamma)x^0(\tau)
+
\Gamma \bar x(\tau),
\end{equation}
where $\Gamma=\operatorname{diag}(\gamma_i)$.

\begin{definition}[Hybrid fast--slow switching system]
\label{def:hybrid_fast_slow_system}

Consider the consensus-induced exit--and--join dynamics
defined by \eqref{eq:degroot_tau_k}--\eqref{eq:intrinsic_update}.
Let $T(\tau)\in\Pi(N)$ denote the coalition structure at
strategic time $\tau$, and let $K_\tau\in\mathbb{N}$
be the corresponding tactical horizon.
Define cumulative switching times
\(T_0:=0\) and \(T_{\tau+1}:=T_\tau+K_\tau\).

The overall evolution of the system is called the
\emph{hybrid fast--slow switching system}
associated with $(V,\{P_S\},\Gamma)$
if it consists of the following coupled updates. For global time
$t\in[T_\tau,T_{\tau+1}-1]$, fast tactical dynamics follow
\begin{equation}
\label{eq:hybrid_fast}
x(t+1)
=
\widetilde P(\tau)\,x(t),
\end{equation}
where $\widetilde P(\tau)$ is the block-diagonal
row-stochastic matrix induced by $T(\tau)$.

Under primitivity of each block $P_S$, coalition-wise consensus satisfies
\begin{equation}
\label{eq:hybrid_consensus_def}
\lim_{t\uparrow T_{\tau+1}}
x(t)
=
\bar x(\tau),
\end{equation}
where $\bar x(\tau)$ stacks the coalition consensus states
defined in \eqref{eq:consensus_limit}.

At switching instants $T_{\tau+1}$, intrinsic opinions evolve according to
\begin{equation}
\label{eq:hybrid_slow}
x^0(\tau+1)
=
(I-\Gamma)x^0(\tau)
+
\Gamma \bar x(\tau),
\end{equation}
which coincides with \eqref{eq:intrinsic_update}
in stacked form.

Given $x^0(\tau+1)$, coalition values are evaluated,
Aumann--Dr\`eze payoffs are computed,
and admissible exit--and--join deviations
determine the next coalition structure $T(\tau+1)$,
which induces the next switching matrix
$\widetilde P(\tau+1)$.

\end{definition}

\begin{proposition}[Sufficient condition for tactical unanimity with reset]
\label{prop:tactical_unanimity_union_connectivity}
Consider the hybrid fast--slow switching system of
Definition~\ref{def:hybrid_fast_slow_system}.
Let $T(\tau)$ denote the coalition structure at strategic time $\tau$,
and let $\mathcal G(\tau)$ denote the directed graph induced by the
block-diagonal matrix $\widetilde P(\tau)$.

Assume that each $\widetilde P(\tau)$ is row-stochastic; that there exists
$\alpha>0$ such that $\widetilde p_{ij}(\tau)\ge\alpha$ whenever
$\widetilde p_{ij}(\tau)>0$; that there exists $L\ge1$ such that, for every
$\tau$, the union graph
\(\bigcup_{s=\tau}^{\tau+L-1}\mathcal G(s)\) is strongly connected; and that
$\gamma_i\in(0,1]$ for all $i\in N$.
Then the intrinsic opinions satisfy
\(\lim_{\tau\to\infty}x_i^0(\tau)=x^\star\) for all \(i\in N\),
for some $x^\star\in\mathbb{R}^d$.
Consequently,
\(\lim_{t\to\infty}x_i(t)=x^\star\) for all \(i\in N\),
and tactical unanimity holds.
\end{proposition}

\begin{proof}
At the end of tactical window $\tau$, the coalition-wise
consensus can be written in stacked form as
\(\bar x(\tau)=\mathcal C(T(\tau))\,x^0(\tau)\),
where
\[
\mathcal C(T(\tau))
=
\operatorname{blkdiag}
\bigl(
\mathbf 1 (\pi^S)^\top
\bigr)_{S\in T(\tau)}
\]
is a block rank-one row-stochastic projection matrix.
The intrinsic update \eqref{eq:intrinsic_update} becomes
\(x^0(\tau+1)=(I-\Gamma)x^0(\tau)+
\Gamma\mathcal C(T(\tau))x^0(\tau)\),
where $\Gamma=\operatorname{diag}(\gamma_i)$.
Define
\(A(\tau):=(I-\Gamma)+\Gamma\mathcal C(T(\tau))\).
Then
\(x^0(\tau+1)=A(\tau)x^0(\tau)\).
Each $A(\tau)$ is row-stochastic and has strictly positive diagonal
entries since $\gamma_i>0$. Moreover, whenever agents $i$ and $j$
belong to the same coalition at time $\tau$, the entry
$A_{ij}(\tau)$ is bounded below by a positive constant depending on
$\alpha$ and $\min_i \gamma_i$.
Because the union of coalition graphs over every window of length $L$
is strongly connected, the same holds for the graphs induced by
$A(\tau)$ (see e.g., \cite{jadbabaie2003coordination}). Uniform positivity then
implies that the products \(A(\tau+L-1)\cdots A(\tau)\) are scrambling
stochastic matrices.
By the classical Wolfowitz--Hajnal theorem for products of
row-stochastic matrices with repeated joint connectivity, the products
\(A(\tau)\cdots A(0)\) converge to a rank-one matrix \(\mathbf 1\pi^\top\)
for some stochastic vector \(\pi\).
Therefore,
\(
x^0(\tau)
\to
(\pi^\top x^0(0))\,\mathbf 1
=
x^\star \mathbf 1,
\)
for some $x^\star\in\mathbb{R}^d$,
which establishes convergence of intrinsic opinions to unanimity.
Within each tactical window,
\(
x(t)
=
\widetilde P(\tau)^{\,t-T_\tau}x^0(\tau).
\)
Since $x^0(\tau)\to x^\star\mathbf 1$ and each
$\widetilde P(\tau)$ is row-stochastic,
it follows that $x(t)\to x^\star\mathbf 1$ as $t\to\infty$.
\end{proof}

\begin{remark}[Strategic unanimity implies tactical unanimity]
If the grand coalition $T(\tau)=\{N\}$ is formed from some
strategic time onward and the corresponding interaction matrix
$P_N$ is strongly connected (equivalently, primitive),
then tactical unanimity follows without requiring
joint connectivity over time.

Indeed, in that case the switching matrix
$\widetilde P(\tau)$ reduces to a single primitive
row-stochastic matrix on $N$.
The tactical dynamics \eqref{eq:degroot_tau_k}
converge to a unique consensus state within each window,
and the intrinsic update preserves unanimity once achieved.
Hence strategic unanimity together with connectivity of the
grand coalition implies tactical unanimity.
\end{remark}

\begin{proposition}[Characterization of the unanimity limit]
\label{prop:limit_characterization}
Under the assumptions of
Proposition~\ref{prop:tactical_unanimity_union_connectivity},
the intrinsic opinions evolve according to
\(x^0(\tau+1)=A(\tau)x^0(\tau)\), with
\(A(\tau)=(I-\Gamma)+\Gamma\mathcal C(T(\tau))\),
where each $A(\tau)$ is row-stochastic.
Then there exists a unique stochastic vector
$\pi\in\Delta^n$ such that
\begin{equation}
\label{eq:rank_one_limit}
\lim_{\tau\to\infty}
A(\tau)A(\tau-1)\cdots A(0)
=
\mathbf 1 \pi^\top.
\end{equation}

Consequently,
\begin{equation}
\label{eq:limit_value}
x^0(\tau)
\longrightarrow
(\pi^\top x^0(0))\,\mathbf 1,
\end{equation}
and the unanimity limit is
\(
x^\star = \pi^\top x^0(0).
\)
The vector $\pi$ depends on the entire switching sequence
$\{T(\tau)\}_{\tau\ge0}$ and the belief gains $\Gamma$.
In general, $x^\star$ need not coincide with the consensus
associated with any fixed coalition structure.
Moreover, the unanimity limit $x^\star$ corresponds to the
intrinsic component of a joint tactical--strategic equilibrium
if and only if there exists a coalition structure
$T^\star \in \Pi(N)$ and a profile
$x^{0,\star}\in\mathbb{R}^{dn}$
such that the strategic structure is eventually stationary, meaning that
there exists $\tau_0$ with \(T(\tau)=T^\star\) for all
\(\tau\ge\tau_0\). The intrinsic profile must satisfy
\(x^{0,\star}=(I-\Gamma)x^{0,\star}+
\Gamma\mathcal C(T^\star)x^{0,\star}\),
equivalently,
\(x_i^{0,\star}=\bar x_{C_{T^\star}(i)}(x^{0,\star})\) for all \(i\in N\).
Finally, $T^\star$ admits no profitable unilateral exit--and--join deviation
under the intrinsic profile $x^{0,\star}$.

In this case,
\(
x^\star \mathbf 1 = x^{0,\star},
\)
and the unanimity value $x^\star$ equals the
coalition consensus associated with $T^\star$.
\end{proposition}

\begin{proof}
Under the assumptions of
Proposition~\ref{prop:tactical_unanimity_union_connectivity},
the matrices $A(\tau)$ are row-stochastic, have uniformly positive
diagonal entries, and satisfy repeated joint strong connectivity.
Moreover, there exists $\beta>0$ such that whenever
$A_{ij}(\tau)>0$, we have $A_{ij}(\tau)\ge \beta$.
By the Wolfowitz--Hajnal theorem \cite{wolfowitz1963products,hajnal1958weak} for products of stochastic matrices
with repeated joint connectivity and uniform positivity,
the infinite product
\(
A(\tau)A(\tau-1)\cdots A(0)
\)
converges to a rank-one stochastic matrix.
Hence there exists a stochastic vector
$\pi\in\Delta^n$ such that
\(
\lim_{\tau\to\infty}
A(\tau)A(\tau-1)\cdots A(0)
=
\mathbf 1 \pi^\top.
\)
Uniqueness of $\pi$ follows from the uniqueness of the
rank-one limit of such products.
Multiplying the limit relation by $x^0(0)$ yields
\(x^0(\tau)=A(\tau-1)\cdots A(0)x^0(0)\to
(\pi^\top x^0(0))\,\mathbf 1\).
This proves \eqref{eq:limit_value} and establishes
\(
x^\star = \pi^\top x^0(0).
\)
We now characterize when this unanimity limit corresponds to a
joint tactical--strategic equilibrium.
($\Rightarrow$)
Suppose $x^\star$ corresponds to the intrinsic component of
a joint tactical--strategic equilibrium.
By definition of equilibrium, the coalition structure must be
stationary from some $\tau_0$ onward; otherwise the hybrid
operator would not be at a fixed point.
Hence there exists $T^\star$ such that
$T(\tau)=T^\star$ for all $\tau\ge\tau_0$.
Since $A(\tau)$ is constant and equal to
\(
A^\star=(I-\Gamma)+\Gamma\mathcal C(T^\star)
\)
for all $\tau\ge\tau_0$,
the limit $x^{0,\star}=x^\star \mathbf 1$
must satisfy the fixed-point condition
\(
x^{0,\star}=A^\star x^{0,\star},
\)
which is equivalent to
\(x_i^{0,\star}=\bar x_{C_{T^\star}(i)}(x^{0,\star})\) for all \(i\).
Strategic stability of $T^\star$ follows from the
equilibrium assumption.
($\Leftarrow$)
Conversely, suppose there exist $T^\star$ and
$x^{0,\star}$ satisfying (i)--(iii).
Then for all $\tau\ge\tau_0$,
\(
A(\tau)=A^\star
=
(I-\Gamma)+\Gamma\mathcal C(T^\star),
\)
and $x^{0,\star}$ satisfies
\(
x^{0,\star}=A^\star x^{0,\star}.
\)
Hence $x^{0,\star}$ is a fixed point of the intrinsic update,
and since $T^\star$ is strategically stable,
the pair $(T^\star,x^{0,\star})$ is a fixed point of the
hybrid operator.
Therefore the intrinsic component equals
$x^\star \mathbf 1$.
\end{proof}

\subsection{Polarization and Segregation}

\subsubsection{Segregation}

\begin{definition}[Consensus-induced game and segregated equilibrium]
Let $N$ be finite. For each coalition $S \subseteq N$, define
\(\bar x_S=\sum_{i\in S}\pi_i^Sx_i^0\),
where $\pi^S$ is the stationary distribution of a primitive
row-stochastic matrix $P_S$.
Let $V:\mathbb{R}^d \to \mathbb{R}$ be $C^2$, and define
\(v(S)=V(\bar x_S)\) and \(v(\emptyset)=0\).

A pair $(T^\star,\{x_i^{0,\star}\})$ is a joint tactical--strategic
equilibrium if (i) each $S\in T^\star$ attains consensus
$\bar x_S^\star$, (ii) no agent has a profitable unilateral
exit--and--join deviation under the Aumann--Dr\`eze payoff,
and (iii) $x_i^{0,\star}=\bar x_{C_{T^\star}(i)}^\star$ for all $i$.

Define the set of equilibrium consensus states
\(
\mathcal X^\star := \{\bar x_S^\star : S\in T^\star\}.
\)
The equilibrium $T^\star$ is said to be \emph{segregated} if
\(
|\mathcal X^\star| \ge 2.
\)
It is \emph{$k$-segregated} if $|\mathcal X^\star| = k \ge 2$.
\end{definition}

\begin{theorem}[Equilibrium integral balance and geometric obstruction]
\label{thm:integral_balance}
Assume $V\in C^1(\mathbb{R}^d)$, switching costs are zero, and acceptance is automatic.
Let $(T^\star,\{x_i^{0,\star}\})$ be a joint tactical--strategic equilibrium. Fix distinct coalitions $S_1,S_2\in T^\star$ and an agent $i\in S_1$.
Define the displacement vector
\(
\Delta := \bar x_{S_1}^\star - \bar x_{S_2}^\star,
\)
and the affine path
\(\gamma(t):=\bar x_{S_2}^\star+t\Delta\), \(t\in[0,1]\).

Let
\[
w_R := \frac{|R|!\,(|S_2|-|R|)!}{(|S_2|+1)!},
\quad
\alpha_{i,R} := \pi_i^{R\cup\{i\}},
\qquad R\subseteq S_2,
\]
and
\[
\tilde w_R := \frac{|R|!\,(|S_1|-|R|-1)!}{|S_1|!},
\quad
\tilde\alpha_{i,R} := \pi_i^{R\cup\{i\}},
\qquad R\subseteq S_1\setminus\{i\}.
\]

Then strategic stability of $T^\star$ implies the integral balance condition
\begin{equation}
\sum_{R\subseteq S_2} w_R
\int_{0}^{\alpha_{i,R}}
\langle \nabla V(\gamma(t)),\Delta\rangle dt
\;\le\;
\sum_{R\subseteq S_1\setminus\{i\}} \tilde w_R
\int_{0}^{\tilde\alpha_{i,R}}
\langle \nabla V(\gamma(1-t)),\Delta\rangle dt.
\label{eq:integral_balance}
\end{equation}

Moreover, if the equilibrium is segregated, i.e.
\(
\bar x_{S_1}^\star \neq \bar x_{S_2}^\star, (\Delta \neq 0),
\)
then there exists $t^\star\in(0,1)$ such that
\begin{equation}
\langle \nabla V(\gamma(t^\star)),\Delta\rangle = 0.
\label{eq:directional_zero}
\end{equation}
\end{theorem}

\begin{proof}
Let $S_1,S_2 \in T^\star$ be distinct coalitions and fix $i\in S_1$.
Define
\(\Delta:=\bar x_{S_1}^\star-\bar x_{S_2}^\star\) and
\(\gamma(t):=\bar x_{S_2}^\star+t\Delta\).
Because switching costs are zero and acceptance is automatic,
agent $i$ cannot profit by moving from $S_1$ to $S_2$:
\begin{equation}
\Omega_i(v^\star;T^\star)
\ge
\Omega_i(v^\star;F_i(T^\star,S_2)).
\label{eq:stab_main}
\end{equation}
Under the Aumann--Dr\`eze value,
$i$'s payoff inside a coalition equals its Shapley value
in the restricted cooperative game.
Let $S_2^+ := S_2 \cup \{i\}$.
For $R\subseteq S_2$, the marginal contribution of $i$ equals
\(
V(\bar x_{R\cup\{i\}}^\star) - V(\bar x_R^\star).
\)
At equilibrium,
\(
\bar x_R^\star = \bar x_{S_2}^\star,
\)
and when $i$ joins $R$,
\(\bar x_{R\cup\{i\}}^\star=\bar x_{S_2}^\star+\alpha_{i,R}\Delta\), where
\(\alpha_{i,R}=\pi_i^{R\cup\{i\}}\).
Hence
\[
V(\bar x_{R\cup\{i\}}^\star) - V(\bar x_R^\star)
=
V(\bar x_{S_2}^\star + \alpha_{i,R}\Delta)
-
V(\bar x_{S_2}^\star).
\]
Applying the fundamental theorem of calculus along $\gamma$,
\[
V(\bar x_{S_2}^\star + \alpha_{i,R}\Delta)
-
V(\bar x_{S_2}^\star)
=
\int_0^{\alpha_{i,R}}
\langle \nabla V(\gamma(t)), \Delta \rangle dt.
\]
Therefore,
\[
\Omega_i(v^\star;F_i(T^\star,S_2))
=
\sum_{R\subseteq S_2}
w_R
\int_0^{\alpha_{i,R}}
\langle \nabla V(\gamma(t)), \Delta \rangle dt.
\]
Inside $S_1$, the Shapley formula gives
\[
\Omega_i(v^\star;T^\star)
=
\sum_{R\subseteq S_1\setminus\{i\}}
\tilde w_R
\Bigl(
V(\bar x_{S_1}^\star)
-
V(\bar x_{S_1}^\star - \tilde\alpha_{i,R}\Delta)
\Bigr).
\]
Note that
\(\bar x_{S_1}^\star-\tilde\alpha_{i,R}\Delta=\gamma(1-\tilde\alpha_{i,R})\).
Again by the fundamental theorem of calculus,
\[
V(\bar x_{S_1}^\star)
-
V(\bar x_{S_1}^\star - \tilde\alpha_{i,R}\Delta)
=
\int_0^{\tilde\alpha_{i,R}}
\langle \nabla V(\gamma(1-t)), \Delta \rangle dt.
\]
Thus
\[
\Omega_i(v^\star;T^\star)
=
\sum_{R\subseteq S_1\setminus\{i\}}
\tilde w_R
\int_0^{\tilde\alpha_{i,R}}
\langle \nabla V(\gamma(1-t)), \Delta \rangle dt.
\]
Substituting the expressions for both payoffs into
\eqref{eq:stab_main}
yields the claimed inequality
\eqref{eq:integral_balance}.
Assume now that the equilibrium is segregated, so that
$\Delta \neq 0$.
Suppose, for contradiction, that
\(\langle\nabla V(\gamma(t)),\Delta\rangle>0\) for all \(t\in[0,1]\).
Then every integral on the left-hand side of
\eqref{eq:integral_balance}
is strictly positive, and so is every integral on the right-hand side.
Moreover, because the Shapley weights form probability measures,
the deviation payoff strictly exceeds the equilibrium payoff,
contradicting \eqref{eq:stab_main}.
Similarly, if the directional derivative were strictly negative
for all $t\in[0,1]$, the reverse deviation would be profitable. Hence there
must exist $t^\star \in (0,1)$ such that
\(
\langle \nabla V(\gamma(t^\star)), \Delta \rangle = 0.
\) 
This establishes \eqref{eq:directional_zero}.
\end{proof}

\begin{corollary}[Segregation under affine performance functionals]
\label{cor:affine_case}

Suppose switching costs are zero and acceptance is automatic.
Let $(T^\star,\{x_i^{0,\star}\})$ be a joint tactical--strategic equilibrium,
and assume that
\(V(x)=a+\langle b,x\rangle\), with \(b\in\mathbb{R}^d\).

Fix distinct coalitions $S_1,S_2\in T^\star$ and define
\[
\Delta:=\bar x_{S_1}^\star-\bar x_{S_2}^\star,
\qquad
\bar\alpha_2:=\sum_{R\subseteq S_2} w_R \alpha_{i,R},
\qquad
\bar\alpha_1:=\sum_{R\subseteq S_1\setminus\{i\}}
\tilde w_R \tilde\alpha_{i,R}.
\]

Then the equilibrium condition \eqref{eq:integral_balance}
reduces to
\(\langle b,\Delta\rangle(\bar\alpha_2-\bar\alpha_1)\le0\).

In particular, if $\langle b,\Delta\rangle=0$
the condition is automatically satisfied,
while if $\langle b,\Delta\rangle\neq0$
stability depends only on the sign of
$\bar\alpha_2-\bar\alpha_1$.
\end{corollary}

\begin{proof}
Since $V$ is affine,
$\nabla V(x)\equiv b$ and $\nabla^2 V\equiv0$.
Thus
\(\langle\nabla V(\gamma(t)),\Delta\rangle=\langle b,\Delta\rangle=:c\)
is constant in $t$.
Each integral in \eqref{eq:integral_balance} therefore equals
$c$ times the corresponding coefficient:
\[
\int_0^{\alpha_{i,R}}
\langle \nabla V(\gamma(t)),\Delta\rangle dt
=
\alpha_{i,R} c,
\qquad
\int_0^{\tilde\alpha_{i,R}}
\langle \nabla V(\gamma(1-t)),\Delta\rangle dt
=
\tilde\alpha_{i,R} c.
\]
Substituting into \eqref{eq:integral_balance}
and factoring out $c=\langle b,\Delta\rangle$
gives
\(\langle b,\Delta\rangle(\bar\alpha_2-\bar\alpha_1)\le0\),
which proves the claim.
\end{proof}

\subsection{Polarization and Performance Landscape}

\begin{definition}[Polarization equilibrium]
Let $(T^\star,\{x_i^{0,\star}\}_{i\in N})$ be a joint
tactical--strategic equilibrium.
For each coalition $S\in T^\star$, define its equilibrium
consensus state
\(
\bar x_S^\star
=
\sum_{i\in S}\pi_i^S x_i^{0,\star}.
\)
Let
\(
\mathcal X^\star
:=
\{\bar x_S^\star : S\in T^\star\}.
\)

The equilibrium is called a \emph{polarization equilibrium} if
$|\mathcal X^\star| \ge 2$ and, for any two distinct coalitions
$S_1,S_2\in T^\star$,
\(v(S_1\cup S_2)<v(S_1)+v(S_2)\).
\end{definition}

\begin{theorem}[Polarization and landscape structure]
\label{thm:polarization_landscape}
Let $(T^\star,\{x_i^{0,\star}\}_{i\in N})$ be a joint
tactical--strategic equilibrium and suppose
$v(S)=V(\bar x_S)$ for some function
$V:\mathbb{R}^d\to\mathbb{R}$.
Let $S_1,S_2\in T^\star$ be distinct coalitions and define
$x_1 := \bar x_{S_1}^\star$ and
$x_2 := \bar x_{S_2}^\star$.
Then \(v(S_1\cup S_2) < v(S_1)+v(S_2)\) is equivalent to
\(V(\bar x_{S_1\cup S_2}^\star)<V(x_1)+V(x_2)\), and also to the existence of
\(\lambda\in(0,1)\) such that
\(V(\lambda x_1+(1-\lambda)x_2)<V(x_1)+V(x_2)\).
Moreover, if $x_1\neq x_2$ at equilibrium,
then $V$ cannot be strictly convex or strictly concave
on $\mathbb{R}^d$.
\end{theorem}

\begin{proof}
Since $v(S)=V(\bar x_S)$, condition (i) is exactly (ii),
so they are equivalent.
Consensus is affine in intrinsic opinions.
Thus there exists $\lambda\in(0,1)$ such that
$\bar x_{S_1\cup S_2}^\star
=\lambda x_1+(1-\lambda)x_2$.
Substituting this identity into (ii) yields (iii),
and conversely.
Hence all three statements are equivalent.
For the final claim, suppose $x_1\neq x_2$.
If $V$ is strictly concave, then for every $\lambda\in(0,1)$,
$V(\lambda x_1+(1-\lambda)x_2)
>
\lambda V(x_1)+(1-\lambda)V(x_2)$,
which contradicts (iii).
If $V$ is strictly convex and both $x_1$ and $x_2$
are equilibrium consensus values (hence local maximizers),
strict convexity implies uniqueness of the maximizer,
so $x_1=x_2$, a contradiction.
Thus polarization requires that $V$
fail to be globally strictly convex or strictly concave.
\end{proof}

\begin{example}[Double-well landscape]
\label{ex:double_well_shifted}
Let \(d=1\) and define
\(V(x)=-\big((x-0.5)^2-0.4^2\big)^2\).
The function has two strict local maxima at \(x_1=0.1\) and \(x_2=0.9\), and a
strict local minimum at \(x=0.5\). Indeed,
\((x-0.5)^2-0.4^2=0\) if and only if \(x=0.5\pm0.4\).
For any \(\lambda\in(0,1)\), the convex combination
\(\lambda x_1+(1-\lambda)x_2=0.9-0.8\lambda\) lies strictly between the two
maximizers, so \(|\lambda x_1+(1-\lambda)x_2-0.5|<0.4\). Consequently,
\(V(\lambda x_1+(1-\lambda)x_2)<V(x_1)=V(x_2)\), and the landscape satisfies the
polarization condition.
\end{example}

\subsubsection{Segregation and Polarization}

Segregation and polarization describe distinct mechanisms
behind multi-cluster equilibria.
An equilibrium is \emph{segregated} if the set of consensus values
$\mathcal X^\star := \{\bar x_S^\star : S\in T^\star\}$
contains at least two distinct elements.
Segregation is therefore a structural property of the equilibrium
partition: different coalitions stabilize at different consensus states.
It does not impose any condition on the geometry of $V$.
Polarization imposes an additional landscape requirement.
For distinct coalitions $S_1,S_2\in T^\star$
with consensus values $x_1$ and $x_2$,
polarization requires
$v(S_1\cup S_2) < v(S_1)+v(S_2)$.
By Theorem~\ref{thm:polarization_landscape},
this is equivalent to the existence of
$\lambda\in(0,1)$ such that
$V(\lambda x_1+(1-\lambda)x_2)
<
V(x_1)+V(x_2)$.
Thus polarization reflects a geometric property of $V$
along the segment connecting cluster consensuses.
Segregation may arise purely from strategic frictions
such as switching costs or acceptance rules,
even when merging would improve aggregate performance.
Polarization, by contrast, requires that the performance landscape
itself discourage merging. Hence polarization implies segregation,
but segregation does not imply polarization.

\subsection{Cognitive Barriers}
\label{subsec:cognitive_barriers}

The previous subsection characterized segregation geometrically,
through directional derivatives of the performance functional $V$
along the line segment connecting equilibrium consensus states.
We now introduce a complementary \emph{payoff-based}
measure of stability.
While equilibrium ensures that no agent has a profitable unilateral deviation,
it does not quantify how robust the segregation is.
An equilibrium may be strictly stable, marginal, or nearly unstable.
To measure the stability margin of a segregated configuration,
we introduce the notion of a \emph{cognitive barrier}:
the minimal compensating incentive required to induce
an agent to abandon its current coalition.
This concept captures rational resistance embedded in the equilibrium.
Even when multiple coalitions coexist,
agents may require strictly positive incentives to reconsider
their alignment.
The magnitude of this resistance quantifies the robustness
of segregation.
Let $(T^\star,\{x_i^{0,\star}\}_{i\in N})$ be a joint
tactical--strategic equilibrium.
For each coalition $S\in T^\star$, denote its equilibrium
consensus state by
\(\bar x_S^\star=\sum_{j\in S}\pi_j^S x_j^{0,\star}\).
Fix distinct coalitions $S,R \in T^\star$.
For agent $i\in S$, define the \emph{switching gain}
\(\Delta_i^{SR}:=\Omega_i\bigl(v;F_i(T^\star,R)\bigr)-\Omega_i(v;T^\star)\),
where $F_i(T^\star,R)$ denotes the partition obtained
by moving $i$ from $S$ to $R$.
The switching gain measures the net payoff advantage
of migrating from $S$ to $R$.
Strategic stability implies
\(\Delta_i^{SR}\le0\) for all \(i\in S\) and \(R\neq S\),
since no profitable unilateral deviation exists.

\begin{definition}[Individual cognitive barrier]
For $i\in S$ and $R\neq S$, define the \emph{cognitive barrier}
\(B_i^{SR}:=-\Delta_i^{SR}=\Omega_i(v;T^\star)-
\Omega_i\bigl(v;F_i(T^\star,R)\bigr)\ge0\).
Equivalently, $B_i^{SR}$ is the smallest transfer
$\tau\ge0$ such that
\(\Omega_i\bigl(v;F_i(T^\star,R)\bigr)+\tau\ge\Omega_i(v;T^\star)\).
\end{definition}

The quantity $B_i^{SR}$ represents the minimal compensating reward
required to induce agent $i$ to abandon its current coalition $S$
and join $R$.
It quantifies the directional stability margin of agent $i$
with respect to migration from $S$ to $R$.
In particular, $B_i^{SR}=0$ if and only if agent $i$
is indifferent between the two coalitions.

\begin{definition}[Directional barriers]
For distinct coalitions $S,R\in T^\star$, define
\(B^{SR}_{\min}:=\min_{i\in S}B_i^{SR}\) and
\(B^{SR}_{\max}:=\max_{i\in S}B_i^{SR}\).
\end{definition}

The quantity $B^{SR}_{\min}$ is the smallest
incentive required to trigger at least one migration
from $S$ to $R$.
The quantity $B^{SR}_{\max}$ is the uniform incentive
required to relocate all members of $S$ to $R$.
Thus $B^{SR}_{\min}$ measures the weakest link
in coalition stability,
while $B^{SR}_{\max}$ measures its strongest resistance.

\begin{definition}[Global cognitive barrier]
The \emph{global cognitive barrier} of the equilibrium $T^\star$ is
\(B_{\min}(T^\star)=\min_{S\neq R}B^{SR}_{\min}\).
\end{definition}

The scalar $B_{\min}(T^\star)$ provides a global
measure of the robustness of segregation.
If $B_{\min}(T^\star)>0$,
the equilibrium partition is \emph{strictly stable}:
every cross-coalition deviation requires
a strictly positive incentive.
If $B_{\min}(T^\star)=0$,
the equilibrium is \emph{marginal}:
at least one agent is indifferent to switching.
If a negative switching gain were to arise
(i.e., $B_i^{SR}<0$),
strategic stability would fail
and segregation would collapse.
Geometrically, let $\Delta = \bar x_S^\star - \bar x_R^\star$
denote the displacement between two equilibrium
consensus states.
When $V$ is differentiable, the switching gain
admits a line-integral representation along the
segment connecting $\bar x_R^\star$ and $\bar x_S^\star$:
\[
\Delta_i^{SR}
=
\int_0^{\alpha_{i}}
\langle
\nabla V(\bar x_R^\star + t\Delta),
\Delta
\rangle dt
-
\int_0^{\tilde\alpha_{i}}
\langle
\nabla V(\bar x_S^\star - t\Delta),
\Delta
\rangle dt,
\]
where $\alpha_i$ and $\tilde\alpha_i$
are influence coefficients determined by the consensus weights.

Thus cognitive barriers are governed by directional
derivatives of $V$ along the segregation vector.
When the directional derivative vanishes along the
connecting segment, barriers shrink.
When $V$ bends sharply against $\Delta$,
barriers grow.

If $V$ is strictly concave,
the directional derivative
$\langle\nabla V(\gamma(t)),\Delta\rangle$
is strictly monotone along the segment.
In this case segregation is incompatible
with strict stability,
and the global barrier \(B_{\min}(T^\star)\) vanishes.

Conversely, large positive barriers require
sufficient curvature of $V$
to penalize motion toward alternative consensus states.
Hence cognitive barriers encode second-order
information about the geometry of the performance landscape.

\section{Numerical Experiments}
\label{sec:numerical_experiments}

This section provides a computational illustration of the mechanisms developed
above. The experiments are not intended as a separate proof of convergence.
Rather, they show how the strategic exit--and--join rule, the tactical
consensus map, switching frictions, and slow belief adaptation interact along
representative trajectories.

\subsection{Baseline protocol}

\begin{example}[Baseline numerical protocol]
\label{ex:baseline_numerics}
The baseline experiment uses a population of $N=10$ agents initialized as
singleton coalitions.
Intrinsic opinions $x_i^0(0)$ are drawn independently from the uniform
distribution on $[0,1]$.
Within-coalition tactical dynamics follow DeGroot averaging with a uniform
interaction matrix, iterated for $K=10$ steps at each strategic time.
Coalition performance is evaluated by the single-peaked score
\(V(x)=-(x-0.5)^2\), whose maximizer is the target consensus value \(x=0.5\).
Strategic updates use Aumann--Dr\`eze payoffs, Pareto non-decrease acceptance
for destination coalitions, and a switching cost \(c\ge0\).
Intrinsic opinions are updated slowly according to
\(x_i^0(\tau+1)=(1-\gamma)x_i^0(\tau)+\gamma\bar x_{C_T(i)}(\tau)\),
with $\gamma=0.3$ fixed throughout the experiments.
All simulations are run for $\tau_{\max}=40$ strategic iterations.
\end{example}

The baseline comparison varies only the switching cost
\(c\in\{0,0.01,0.02\}\). The plotted diagnostics are the coalition-label
trajectory, the aggregate coalition score
\(W(\tau):=\sum_{S\in T(\tau)}v(S;\tau)\), and the intrinsic-opinion
trajectory. Together these quantities separate the strategic state, the
coalitional objective, and the slow belief state.

\subsection{Nonnegative switching costs}

Figure~\ref{fig:coalition_evolution_panels} shows the evolution of coalition
membership under the three nonnegative switching-cost regimes. Rows correspond
to agents, columns to strategic time, and colors indicate coalition labels.

\begin{figure}[t]
\centering
\begin{subfigure}{0.32\textwidth}
    \centering
    \includegraphics[width=\linewidth]{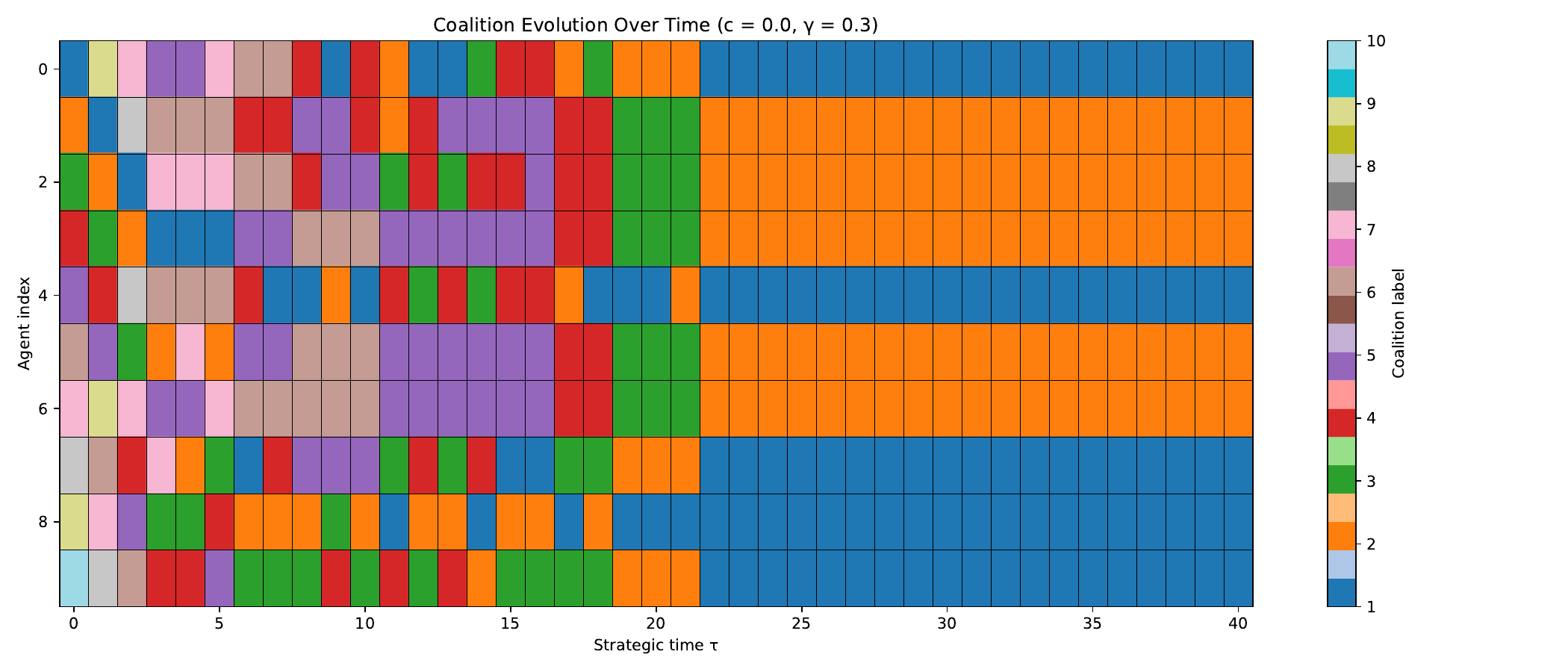}
    \caption{$c = 0$}
\end{subfigure}
\hfill
\begin{subfigure}{0.32\textwidth}
    \centering
    \includegraphics[width=\linewidth]{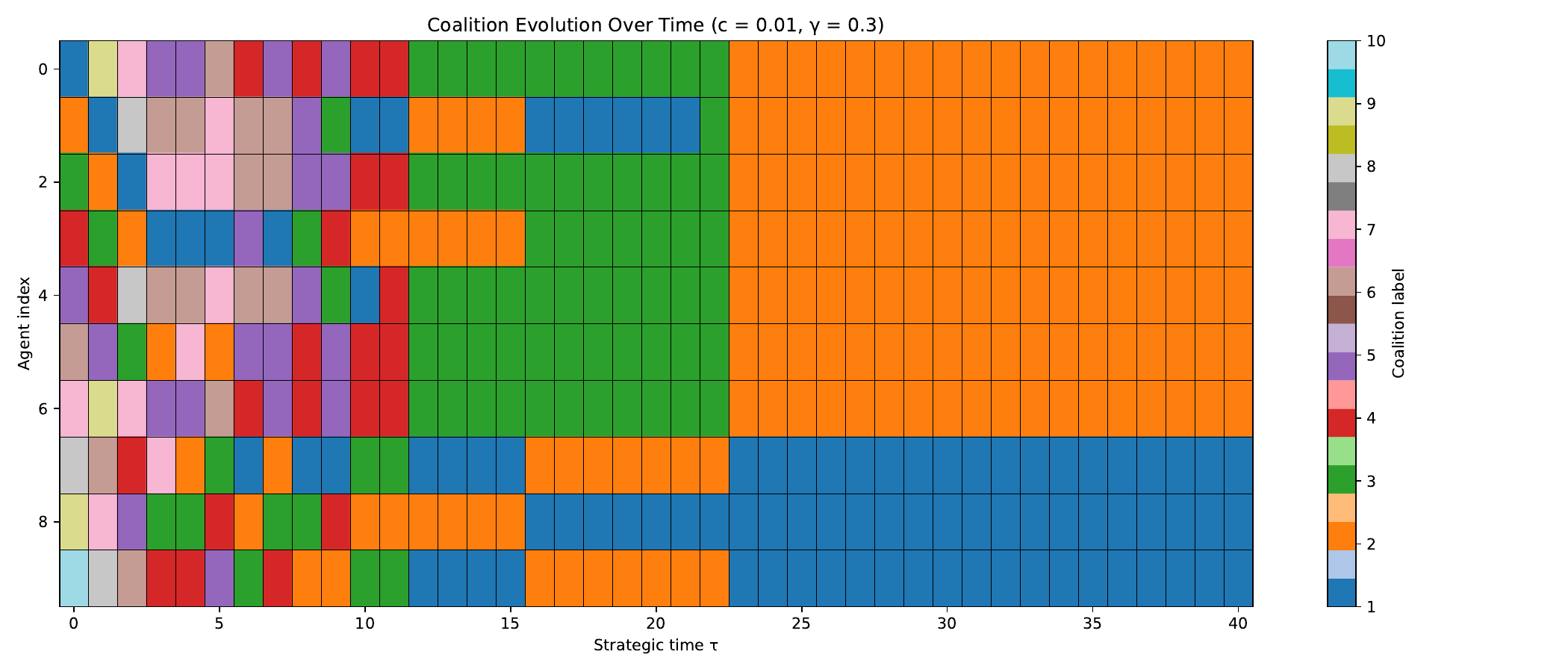}
    \caption{$c = 0.01$}
\end{subfigure}
\hfill
\begin{subfigure}{0.32\textwidth}
    \centering
    \includegraphics[width=\linewidth]{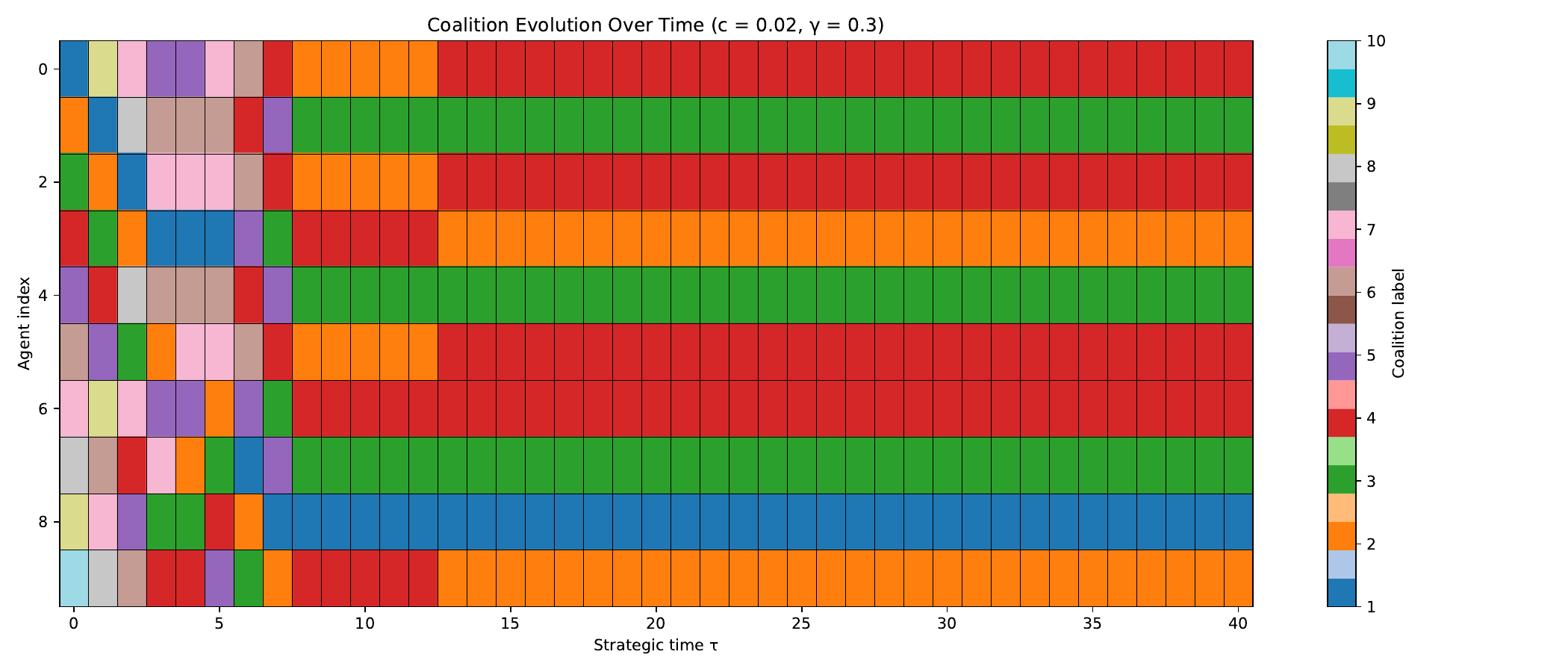}
    \caption{$c = 0.02$}
\end{subfigure}
\caption{Coalition evolution under increasing switching costs
for the baseline protocol in Example~\ref{ex:baseline_numerics}.}
\label{fig:coalition_evolution_panels}
\end{figure}

The costless regime exhibits frequent early reconfiguration: agents enter and
leave coalitions until the payoff-improving moves are exhausted. Positive
switching costs shrink the admissible transition set because a move must now
generate enough Aumann--Dr\`eze improvement to offset the friction. As \(c\)
increases, the observed partitions become more persistent and the transient
reorganization phase shortens. This is the computational analogue of the
cognitive-barrier interpretation in Subsection~\ref{subsec:cognitive_barriers}: a
larger switching cost raises the payoff threshold required to cross from one
coalitional arrangement to another.

Figure~\ref{fig:lyapunov_panels} reports the total coalition surplus
\(W(\tau)=\sum_{S\in T(\tau)}v(S;\tau)\) as a function of strategic time.
The curves are best read as an empirical potential diagnostic for the displayed
trajectories. Implemented moves are payoff-improving after switching costs and
must be accepted by the destination coalition; in these simulations, the
aggregate score increases during the active reconfiguration phase and then
plateaus once the partition becomes strategically stable.

\begin{figure}[t]
\centering
\begin{subfigure}{0.32\textwidth}
    \centering
    \includegraphics[width=\linewidth]{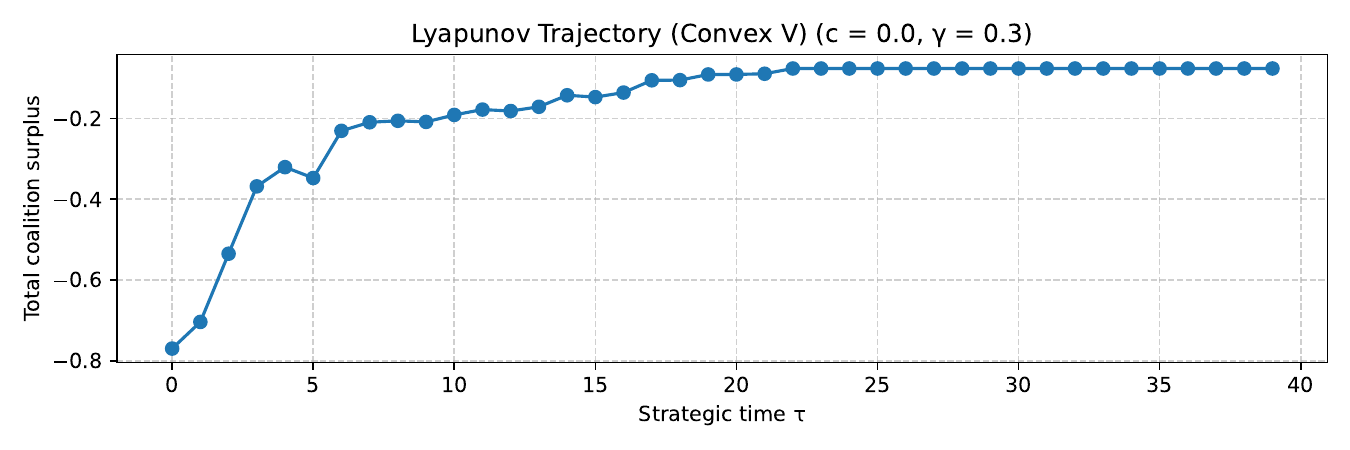}
    \caption{$c = 0$}
\end{subfigure}
\hfill
\begin{subfigure}{0.32\textwidth}
    \centering
    \includegraphics[width=\linewidth]{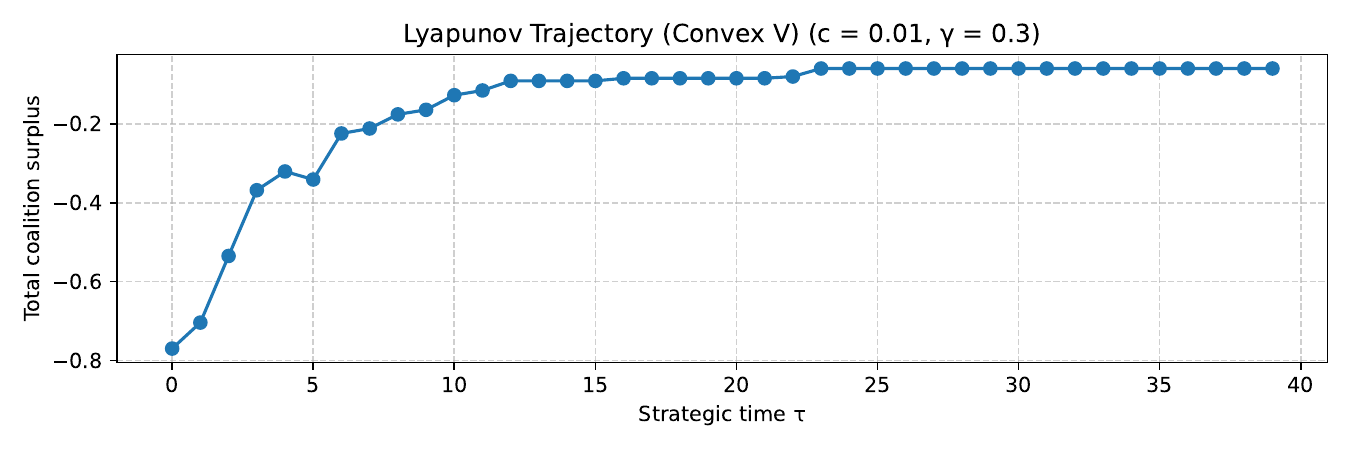}
    \caption{$c = 0.01$}
\end{subfigure}
\hfill
\begin{subfigure}{0.32\textwidth}
    \centering
    \includegraphics[width=\linewidth]{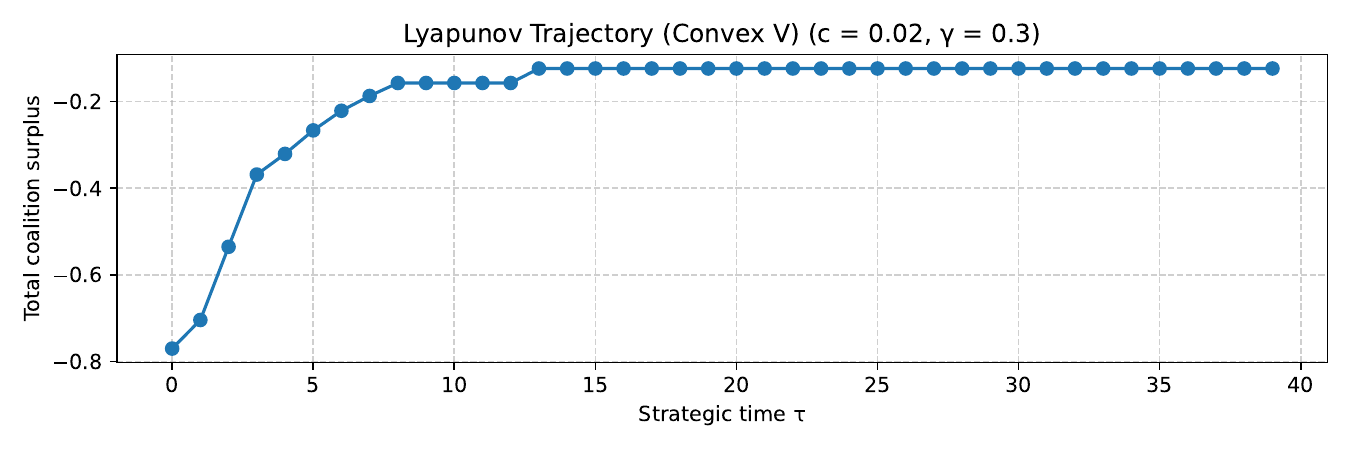}
    \caption{$c = 0.02$}
\end{subfigure}
\caption{Aggregate coalition score \(W(\tau)=\sum_{S\in T(\tau)}v(S;\tau)\)
under nonnegative switching costs.}
\label{fig:lyapunov_panels}
\end{figure}

Higher switching costs reduce the number of strategic moves and therefore
produce flatter surplus trajectories. This does not mean that the tactical
layer is inactive; rather, after the strategic partition stabilizes, additional
tactical averaging occurs within a fixed coalitional architecture.
Figure~\ref{fig:opinion_panels} depicts the evolution of intrinsic opinions for
all agents. Despite heterogeneous initial conditions and changing coalition
memberships, opinions contract toward a common range because each slow update
pulls an agent toward the consensus state of its current coalition.

\begin{figure}[t]
\centering
\begin{subfigure}{0.32\textwidth}
    \centering
    \includegraphics[width=\linewidth]{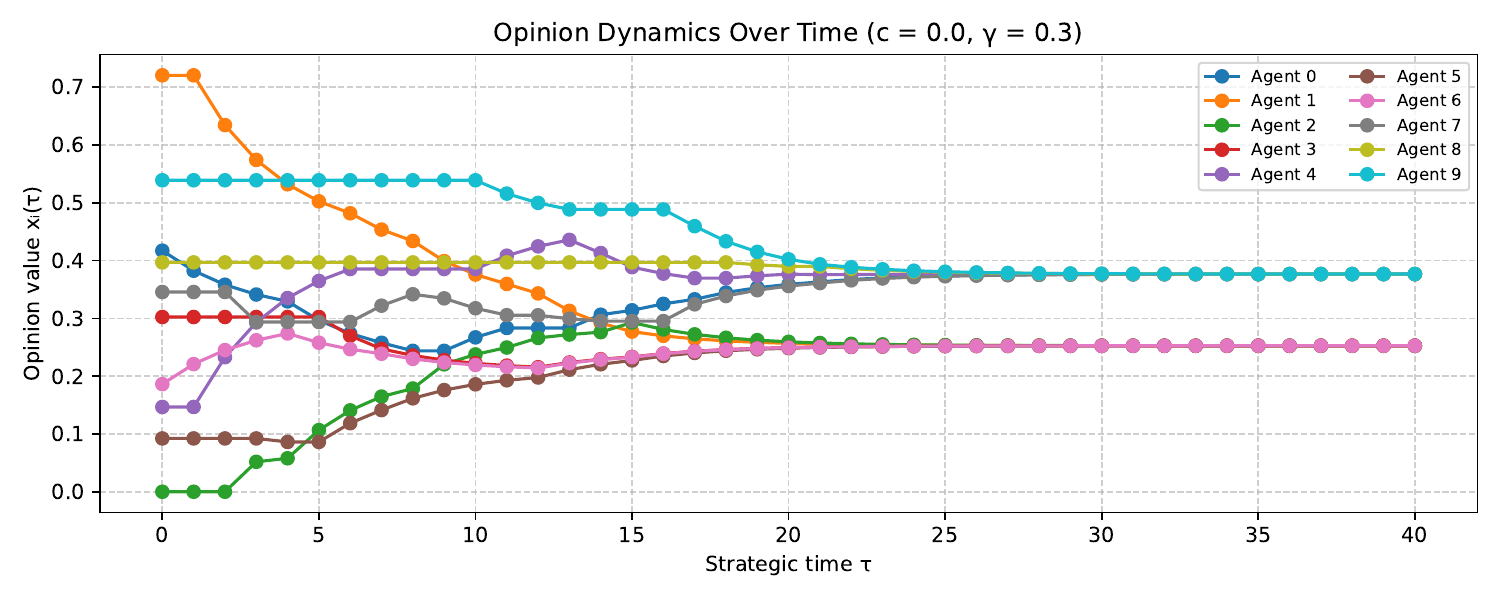}
    \caption{$c = 0$}
\end{subfigure}
\hfill
\begin{subfigure}{0.32\textwidth}
    \centering
    \includegraphics[width=\linewidth]{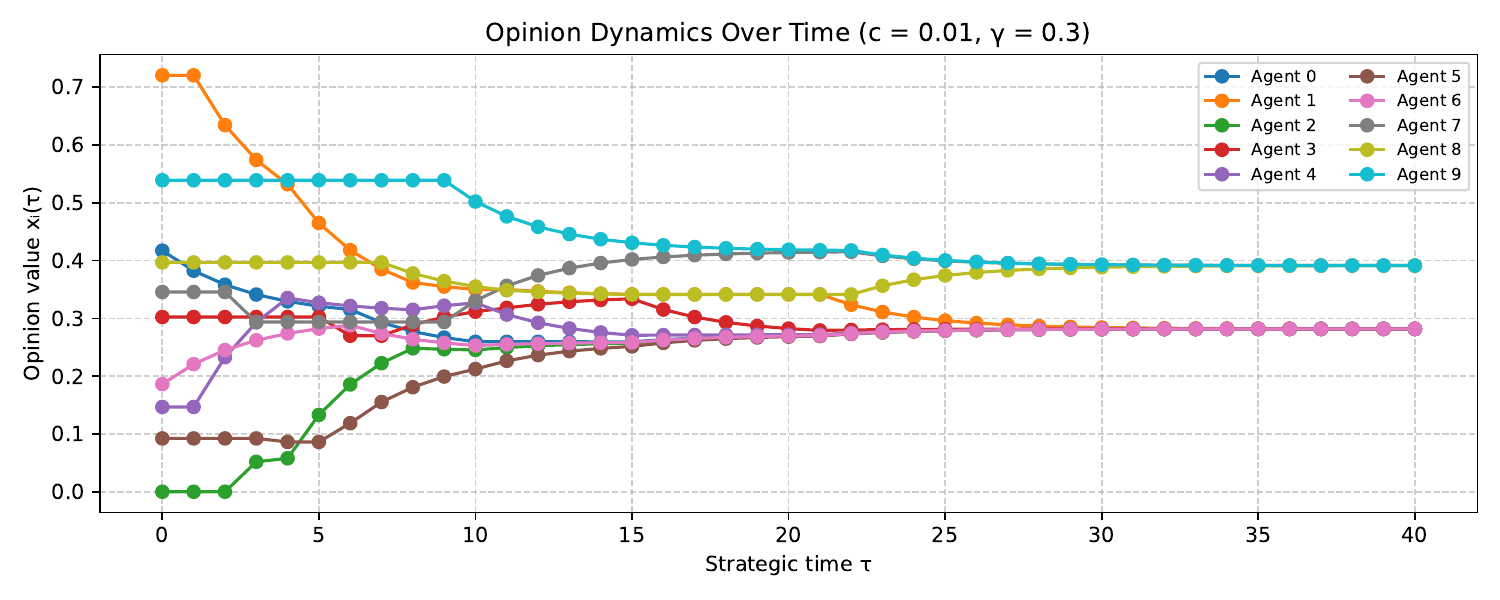}
    \caption{$c = 0.01$}
\end{subfigure}
\hfill
\begin{subfigure}{0.32\textwidth}
    \centering
    \includegraphics[width=\linewidth]{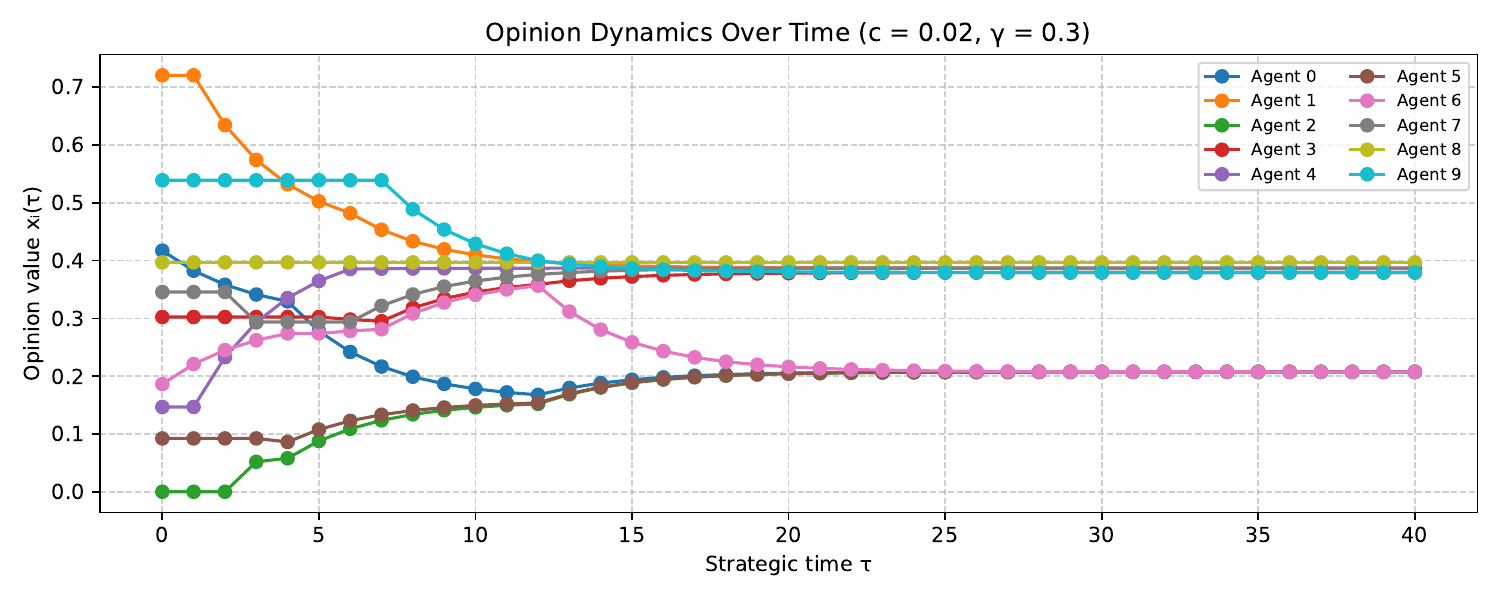}
    \caption{$c = 0.02$}
\end{subfigure}
\caption{Intrinsic opinion trajectories under slow belief adaptation in the
baseline protocol.}
\label{fig:opinion_panels}
\end{figure}

The opinion panels clarify the fast--slow structure of the model. Tactical
consensus acts inside the currently formed coalitions, while the slow update
transfers coalition-level information back into intrinsic opinions. Larger
switching costs indirectly slow global opinion mixing because agents encounter
fewer coalition partners, but the baseline regimes still display contraction of
opinion dispersion.

\subsection{Mobility subsidy and instability-induced mixing}
\label{subsec:negative_cost_exploration}

The preceding experiments restrict \(c\) to be nonnegative, so switching is
never subsidized. We next reverse the sign of this friction and set
\(c_{\mathrm{switch}}=-0.01\), with \(\gamma=0.6\) and the nonconcave performance landscape
\(V(x)=-(x^2-0.5)^2\). A negative switching cost is a mobility subsidy. The
individual improvement condition becomes
\(\Omega_i(T_{\mathrm{new}})>\Omega_i(T)+c_{\mathrm{switch}}
=\Omega_i(T)-0.01\), so a move may be implemented even when it slightly lowers
the agent's immediate Aumann--Dr\`eze payoff. The admissible transition set is
therefore enlarged, and the strategic layer acquires an exploratory component.

Figure~\ref{fig:coalition_negative_cost} displays the coalition evolution.
Coalition labels fluctuate persistently over time and agents repeatedly
enter and leave different coalitions; no stationary partition emerges.

\begin{figure}[t]
\centering
\includegraphics[width=\linewidth]{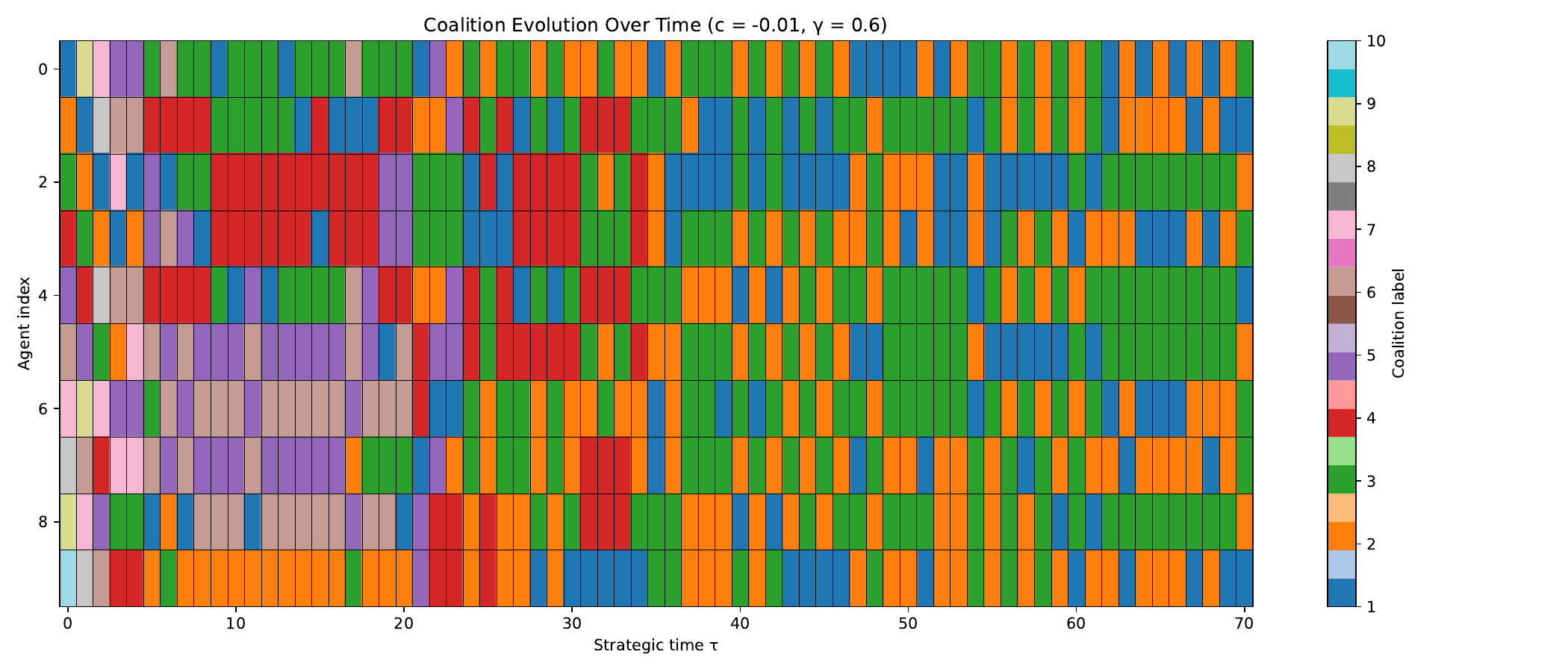}
\caption{Coalition evolution under the mobility subsidy
\(c_{\mathrm{switch}}=-0.01\) and \(\gamma=0.6\). Strategic reconfiguration
persists rather than settling into a fixed partition.}
\label{fig:coalition_negative_cost}
\end{figure}

\begin{figure}[t]
\centering
\includegraphics[width=\linewidth]{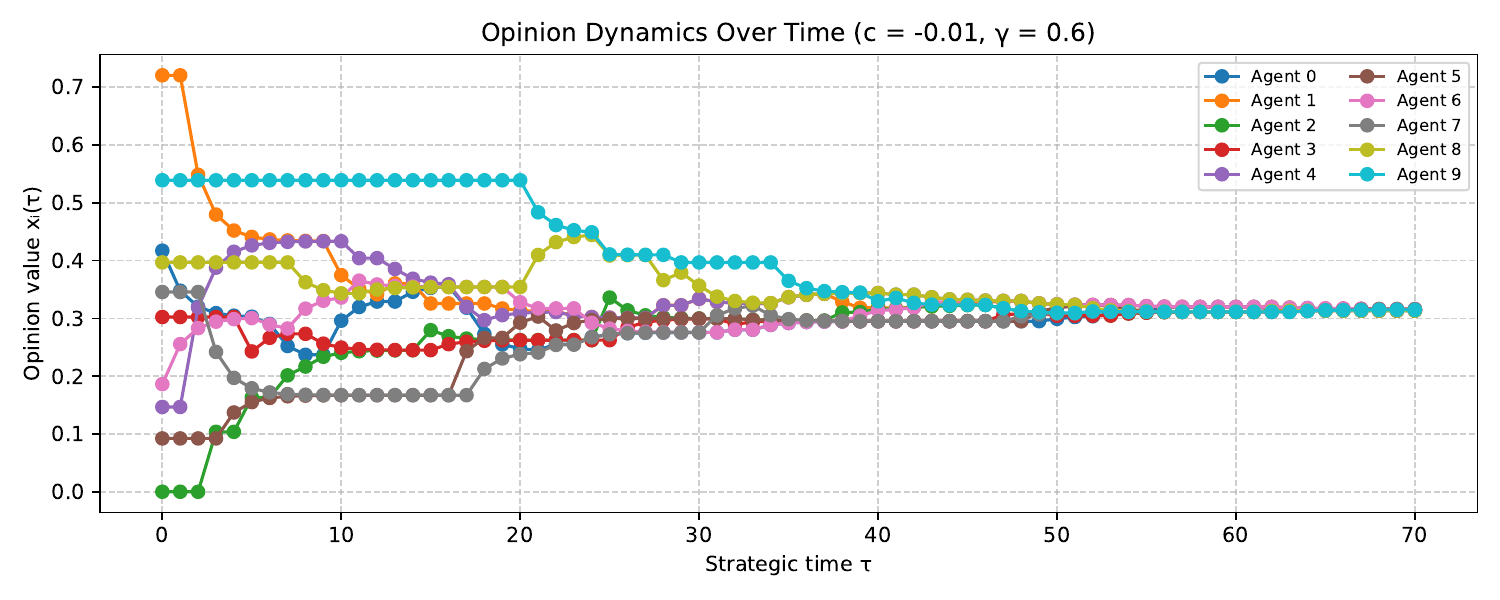}
\caption{Opinion evolution under the mobility subsidy. Despite strategic
nonconvergence, intrinsic opinions converge through repeated coalition-level
mixing.}
\label{fig:opinion_negative_cost}
\end{figure}

The contrast between Figures~\ref{fig:coalition_negative_cost} and
\ref{fig:opinion_negative_cost} is the main numerical point of this regime.
The coalition partition does not converge, yet the intrinsic opinions approach
a common value. This is not a contradiction. At each strategic time, the
partition \(T(\tau)\) induces an averaging operator
\(A(\tau)=(I-\Gamma)+\Gamma\mathcal C(T(\tau))\), so
\(x^0(\tau+1)=A(\tau)x^0(\tau)\). A fixed partition may average only within
disconnected coalitions, but persistent reshuffling changes the interaction
graph over time. When the union of these graphs is repeatedly connected, the
standard consensus mechanism for products of stochastic matrices applies, and
the intrinsic profile contracts toward unanimity.

\begin{observation}[Instability--consensus separation]
\label{obs:instability_consensus_numerical}
In the subsidized-switching regime, strategic nonconvergence can coexist with
tactical convergence. The negative switching cost disrupts the monotone
improvement structure of the coalition process, but the induced sequence of
coalition memberships increases temporal connectivity. In this sense,
strategic instability becomes a mechanism for global opinion mixing rather than
an obstruction to consensus.
\end{observation}

The numerical experiments therefore support the conceptual separation developed
in the theory. Nonnegative switching costs operate as cognitive barriers that
stabilize coalition structures and slow cross-coalition information exchange.
Negative switching costs remove this barrier, producing persistent strategic
motion; nevertheless, the tactical layer can still converge because repeated
coalition reshuffling supplies the temporal connectivity required for global
averaging.

\section{Conclusion}
\label{sec:conclusion}

This paper has developed a multiscale model of coalition formation in which
strategic exit--and--join decisions are coupled to tactical consensus dynamics
inside coalitions. The framework treats coalition value as an endogenous object:
agents first aggregate information through within-coalition DeGroot dynamics,
and the resulting consensus state determines the transferable utility that
drives Aumann--Dr\`eze payoffs, acceptance decisions, and switching incentives.
This construction connects two levels of collective behavior that are often
studied separately. The tactical layer describes how agents coordinate and
produce a coalition-level state, while the strategic layer describes how agents
reconfigure the coalitional architecture in response to the payoff implications
of that state.

The analysis identifies conditions under which the coupled process admits
stable configurations and clarifies when tactical agreement and strategic
coalition formation reinforce or obstruct one another. Switching costs play a
particularly central role. When they are nonnegative, they act as cognitive
barriers that restrict admissible moves, enlarge viability regions, and promote
strategic persistence. Under additional convexity or supermodularity conditions,
the induced coalition game supports stronger forms of strategic unanimity, while
nonconcave performance landscapes and acceptance constraints can sustain
segregation, local stability, or persistent disagreement. These results show
that consensus at the tactical level is neither automatic nor equivalent to
grand-coalition formation at the strategic level; the two notions of unanimity
are linked through the induced value function, payoff allocation rule, and
frictions governing coalition mobility.

The numerical experiments illustrate the same separation in finite
trajectories. Increasing nonnegative switching costs stabilizes coalition
membership and slows cross-coalition information exchange, whereas a mobility
subsidy can destroy strategic convergence while still promoting global tactical
consensus through repeated temporal mixing. This instability--consensus
separation suggests that strategic volatility is not always detrimental: when it
creates sufficiently rich interaction patterns, it can improve information
aggregation even in the absence of a stable partition. Future work may extend
the model to heterogeneous learning rates, stochastic acceptance rules,
network-constrained mobility, and mechanism-design problems in which switching
frictions or subsidies are chosen to balance institutional stability against
global information integration.

\bibliographystyle{abbrv}
\bibliography{ref}

\end{document}